\documentclass{article}

\usepackage{amsmath}
\usepackage{amssymb}
\usepackage{graphicx}

\usepackage{amsthm}
\theoremstyle{definition}
\newtheorem{theo}{Theorem}
\newtheorem{prop}{Proposition}
\newtheorem{lemm}{Lemma}
\newtheorem{fact}{Fact}
\newtheorem{rem}{Remark}
\newtheorem{rei}{Example}
\newtheorem{defi}{Definition}

\def\halmos{$\hfill \square$}

\begin{document}

\title{Characterizations of Controlled Generation of
Right Linear Grammars with Unknown Behaviors}


\author{Daihei Ise\footnote{
    Graduate School of Informatics and Engineering, The University of Electro-Communications}
  \and Satoshi Kobayashi\footnotemark[1]}


\date{}

\maketitle

\begin{abstract}
This paper deals with the control generation of right linear grammars
with unknown behaviors (RLUBs, for short) in which
derivation behavior is not determined completely.
In particular,
we consider a physical property of control devices
used in control systems and formulate it as a partial
order over control alphabet of the control system.
We give necessary and sufficient conditions
for given finite language classes to be generated by
RLUBs and their control systems using a given partial order
over control alphabet. 
\end{abstract}



\section{Introduction}
\label{sec:introduction}

Many molerular computing paradigms have been proposed and studied.
Among them, H systems (\cite{Tom}), P systems (\cite{Paun2}), and
R systems (\cite{Rozenberg2}) are monuments of the theoretical
works leading the molecular computing theory. 
From experimental point of view, 
nucleic acids are materials which are suitable for implementing information processing since the hybridization according to Watson-Crick base pairing can be utilized to encode programs into base sequences (\cite{Adleman,Tom}), 
and various DNA computers have been proposed
(\cite{Tom,Qian1,Qian2,Hagiya1,CRN,Winfree}). 
On the other hand, significant progress has been made in the technology
for controlling DNA hybridization by photo irradiation
(\cite{Fujimoto4,Fujimoto1}), 
temperature change (\cite{Komiya2}), etc. 
These technologies have been applied
to design photo responsive or temperature dependent
DNA devices. 
The progress in DNA nanotechnology and in the control of DNA hybridization 
poses a question whether we can construct a universal system for 
generating a desired DNA nano-structure by controlling with a sequence of
external signals (such as temperature change, or photo irradiation, etc.).
Kimoto, et al., proposed a grammatical system to model such a universal system
for generating linear nano-structures (\cite{Kimoto,Kimoto2}). 
They proposed to use right linear grammars for modelling the control process of generating linear structures,
inspired from the work by Winfree (\cite{Winfree}) 
on the relationship between the class of formal grammars and the class 
of DNA nano-structures, and also from 
many important works on regulated rewriting theories
(\cite{Paun}). Although there have been some works
on the control of R systems (\cite{Petre6,Yako}),
this paper deals with the control of generative process of
formal grammars and we have interests in generating linear structures. 

The model proposed by Kimoto, et al. (\cite{Kimoto,Kimoto2}), 
called {\em right linear grammars with unknown behaviors} (RLUBs, for short), 
is defined as $H=(G,M)$, where $G$ is a right linear grammar, and $M$ is a generative condition which is closely
related to the length of reaction time spent for the generation of
linear nano-structures when we implement $H$ using chemical reactions.
A generative condition is defined as a pair $M=(\mu_t,\mu_b)$, 
where $\mu_t$ and $\mu_b$ are intervals of integers satisfying 
$\mu_b\subseteq\mu_t$. 
Intuitively speaking, $\mu_t$ specifies the set of integers representing 
depth of derivations of $G$ which may possibly occur (upper bound).
On the other hand, $\mu_b$ specifies the set of integers representing
depth of derivations of $G$ which are guaranteed to occur (lower bound).

The control system $C$ for $H$ is defined as a triple $C=(\Gamma,\phi,T)$,
where $\Gamma$ is a finite set of control symbols (each corresponding to
temperature, wave length of light, etc.), $\phi$ is a control function,
and $T$ is a set of strings over $\Gamma$. A control symbol $t\in \Gamma$
activates a specified set $\phi(t)$ of production rules of $G$ of $H$.
However, the invoked set $\phi(t)$ does not determine complete behavior
of $H$ under the control of $C$. We only know the upper bound and
the lower bound of the behavior of $H$. 
$H$ may take any behavior between
the upper and lower bound defined by $\phi(t)$ and $M=(\mu_t,\mu_b)$.
This sort of incompleteness of the knowledge about the behavior of $H$
is motivated by the fact that it is impossible to predict the
behavior of chemical reaction systems completely.

In \cite{Kimoto,Kimoto2}, 
Kimoto, et al., mainly discussed the problem
of controlled generation of a {\em target string} using the framework of RLUBs. Therefore, there has been no general discussion about 
which language classes can be generated by the control of RLUBs
under various generative conditions.
In this paper, we will give necessary and sufficient conditions
of finite language classes to be generated by the control of RLUBs.
The results of this paper could be an important progress for 
future research topics on general theory of controlled generation
of RLUBs.

Furthermore, this paper also extends the notion of monotone property
of control systems in the following way. 
Let us consider temperature dependent DNA devices $M_1$, $M_2$, and $M_3$ such that
$M_i$ is activated at temperature $T_i$ ($i=1,2,3$) or bellow
where $T_1<T_2<T_3$ holds.
Then, the devices activated at temperature $T_i$ are
activated also at temperature $T_j$ if $T_j\le T_i$.
This sort of physical constraint is called {\em monotone property}
in \cite{Kimoto,Kimoto2}.
Let us consider additional photo responsive DNA devices 
$D_1$, $D_2$, and $D_3$ such that $D_i$ is activated by photo irradiation with wave length of $\lambda_i$ nm ($i=1,2,3$) 
or shorter, where $\lambda_1<\lambda_2<\lambda_3$ holds.
In case that we use both of temperature dependent and photo responsive devices, 
a control symbol can be formulated by a pair of solution temperature and
wave length of photo irradiation.
Thus, we have 9 control symbols
$(T_i,\lambda_j)$ for $1 \leq i,j \leq 3$.
For instance, the control symbol
$(T_2,\lambda_1)$ activates DNA devices $M_2$, $M_3$, $D_1$, $D_2$, and $D_3$
and $(T_1,\lambda_2)$ activates $M_1$, $M_2$, $M_3$, $D_2$, and $D_3$.
Let $S(i,j)$ be the set of DNA devices activated by the control symbol
$(T_i,\lambda_j)$. Then, we have that $S(i_1,j_1)\subseteq S(i_2,j_2)$
holds if and only if $i_2 \leq i_1$ and $j_2 \leq j_1$ hold.
Figure \ref{fig1} shows set inclusion relations over the elements
in $\{\,S(i,j)\,\mid\, 1 \leq i,j \leq 3\,\}$,
where
only the relations among the {\em adjacent} elements are shown\footnote{
Here, we mean that $S(i_1,j_1)$ and
$S(i_2,j_2)$ are adjacent if $|i_1-i_2|+|j_1-j_2|=1$ holds.}. 
Therefore, the set inclusion relation of DNA devices
relative to 9 control symbols
could be a partial order.
This paper gives general theoretical analysis for any given partial order
imposed on DNA devices.

The rest of this paper is organized as follows.
Section \ref{sec:preliminaries} introduces essential definitions and notations.
Section \ref{sec:rlub-and-its-control} introduces the definition of a right linear grammar with unknowun behavior (RLUB) and its control system based on \cite{Kimoto,Kimoto2} and also introduces new definition related to RLUBs which were not defined in \cite{Kimoto,Kimoto2}.
Section \ref{sec:Important-Definitions-and-Lemma} introduces the important definitions used for the characterization and shows some propositions and a lemma about the definitions.
Section \ref{sec:Characterization-of-Controlled-Generation-of-RLUBs} shows the characterization of the controlled generation of RLUBs in the erasing mode.
Section \ref{sec:Conclusions} describes the conclusion and the future work.

\begin{figure}
  \begin{center}
    \begin{tabular}{ccccc}
      S(1,1) & $\supseteq$ & S(1,2) & $\supseteq$ & S(1,3) \\
      \rotatebox[origin=c]{90}{$\subseteq$} &  & \rotatebox[origin=c]{90}{$\subseteq$} &  & \rotatebox[origin=c]{90}{$\subseteq$} \\
      S(2,1) & $\supseteq$ & S(2,2) & $\supseteq$ & S(2,3) \\
      \rotatebox[origin=c]{90}{$\subseteq$} &  & \rotatebox[origin=c]{90}{$\subseteq$} &  & \rotatebox[origin=c]{90}{$\subseteq$} \\
      S(3,1) & $\supseteq$ & S(3,2) & $\supseteq$ & S(3,3) \\
    \end{tabular}
    \caption{Set inclusion relations of $S(i,j)$ \label{fig1}}
  \end{center}
\end{figure}

\section{Preliminaries}
\label{sec:preliminaries}

We introduce necessary definitions and notations
based on \cite{Kimoto,Kimoto2}.

Let $\mathbb{Z}$, $\mathbb{Z}_{\geq 0}$
and $\mathbb{Z}_{\geq 1}$ be the sets of integers, non-negative integers,
and positive integers, respectively.
For $i,j \in \mathbb{Z}_{\geq 0}$ with $i \leq j$,
$[i,j]$ is called an {\em interval}
and defined as the finite set
$\{ k \in \mathbb{Z}_{\geq 0} \mid i \leq k \leq j  \}$.
The number of elements of a finite set $X$ is denoted by $|X|$.
The set of all subsets of $X$ is called
the {\em power set of $X$} and is denoted by $2^{X}$.



A finite and non-empty set of symbols is called a {\em finite alphabet}.
Let $\Sigma$ be a finite alphabet.
The length of a string $w$ over $\Sigma$
is denoted by $|w|$.
An {\em empty string} is a string of length 0, and is denoted by $\epsilon$.
The set of all strings over $\Sigma$ is denoted by ${\Sigma}^{*}$.
We define
${\Sigma}^{+} \overset{\mathrm{def}}{=} {\Sigma}^{*} - \{ \epsilon \}$.
A subset of ${\Sigma}^{*}$ is called a {\em language over $\Sigma$}.
A set of languages over $\Sigma$ is called
a {\em class of languages over $\Sigma$}.
A sequence $a a \cdots a \; (a \in \Sigma)$ of length $k$ is denoted by $a^k$, where $k$ is an integer in $\mathbb{Z}_{\geq 0}$.
For a string $w$ and an integer $i$ with $1\le i\le|w|$,
the $i$-th symbol of a sequence $w$ is denoted by $w[i]$.
For a string $w$ and integers $i,j$ with $1 \leq i \leq j \leq |w|$,
by $w[i,j]$, we denote the substring of $w$
starting from the $i$-th symbol $w[i]$ and ending at the $j$-th symbol $w[j]$.
For a finite alphabet $\Sigma$ and an interval $\mu$,
we define
$\Sigma^{\mu} \overset{\mathrm{def}}{=}
\{ w \in {\Sigma}^{*} \mid |w| \in \mu \}$.
For $x,y,z \in \Sigma^*$ such that $xz=y$,
we say that $x$ is a {\em prefix of} $y$.
For $x,y \in \Sigma^*$ and $z \in \Sigma^+$ such that $xz=y$,
we say that $x$ is a {\em proper prefix of} $y$.
For a language $L$ over $\Sigma$, we define a set $\mathrm{Alph}(L)$ as
a set of symbols used in strings in $L$.
Formally, we define $\mathrm{Alph}(L) \overset{\mathrm{def}}{=} \{ a \in \Sigma \mid \exists w \in L \; \exists x,y \in \Sigma^* \mbox{ such that } xay=w \}$.
For a class $\mathcal{L}$ of languages, we define $\mathrm{Alph}(\mathcal{L}) \overset{\mathrm{def}}{=} \bigcup_{L \in \mathcal{L}} \mathrm{Alph}(L)$.

A {\em right linear grammar} is a 4-tuple $G=(V,\Sigma,S,P)$, where
$V$ and $\Sigma$ are finite alphabets such that $V \cap \Sigma = \emptyset$,
$S \in V$ is a special symbol, called a {\em start symbol}, and
$P$ is a finite set of {\em production rules} of the form $A \to aB$, $A \to a$ or $A \to \epsilon \;(A,B\in V, a \in \Sigma)$.
An element of $V$ is called a {\em nonterminal symbol} or a {\em nonterminal}, and an element of $\Sigma$ is called
a {\em terminal symbol} or a {\em terminal}.

Let $G=(V,\Sigma,S,P)$ be a right linear grammar.
For $x,y \in {(V \cup \Sigma)}^*$ and $r:\alpha \to \beta \in P$,
we write $x \underset{r}{\Rightarrow} y$
if there exist $z_1, z_2 \in {(V \cup \Sigma)}^*$ such that
$x = z_1 \alpha z_2$ and $y = z_1 \beta z_2$.
For $x,y \in {(V \cup \Sigma)}^*$, $n \in \mathbb{Z}_{\geq 1}$, and $r_1 , r_2 , \ldots , r_n \in P$,
we write $x \underset{r_1 r_2 \cdots r_n}{\Rightarrow} y$
if there exist $x_0, x_1 , \ldots, x_n \in {(V \cup \Sigma)}^*$ such that
$x=x_0 \underset{r_1}{\Rightarrow} x_1 \underset{r_2}{\Rightarrow} \cdots \underset{r_n}{\Rightarrow} x_n = y$.
For $x \in {(V \cup \Sigma)}^*$ and $\epsilon$,
we write $x \underset{\epsilon}{\Rightarrow} x$.
A subset of $P^*$ is called a {\em behavior of $G$}.
For $x,y \in {(V \cup \Sigma)}^*$ and a behavior $R$ of $G$,
we write $x \underset{R}{\Rightarrow} y$
if there exists $\alpha \in R$ such that $x \underset{\alpha}{\Rightarrow} y$.
The set $L(G) = \{ w \in {\Sigma}^* \mid S \underset{P^*}{\Rightarrow} w  \}$ is called a {\em language generated by $G$}.
Let $i \in \mathbb{Z}_{\geq 1}$.
We define a set $\mathrm{prf}_i (R)$ as
a set of prefixes of length $i$ found in $R$.
For $\alpha \in P^*$, we simply write $\mathrm{prf}_i (\alpha)$ instead of $\mathrm{prf}_i (\{ \alpha \})$.

Let $X$ be a set and $n$ be a positive integer.
By $X^n$, we denote the $n$-term Cartesian product of $X$.
An element of $X^n$ can be regarded as a sequence of
elements of $X$ of length $n$.
For instance, for an interval $\mu$ and a positive integer $n$,
$\lambda \in \mu^n$ implies that $\lambda$ is a sequence of
integers in the set $\mu$ of length $n$.
For sequences of elements in $X$,
we use the same conventional notations as those used for strings.
By $X^+$, we denote the set of sequences of elements in $X$
of finite length at least $1$.
For a sequence $\lambda \in X^+$ and $i \in Z_{\geq 1}$,
$|\lambda|$ and $\lambda[i]$ are defined in the same way
as $|w|$ and $w[i]$ for a string $w$.
The only  notational difference between sequences of
elements in $X$ and strings is
that $\lambda \in X^+$ is represented as a bracketed
comma delimited sequence
$\lambda=<x_1, \ldots x_k> \; (x_1,\ldots,x_k \in X)$,
although a string over a finite alphabet $\Sigma$ is
usually written as $a_1 \cdots a_k \; (a_1,\ldots,a_k \in \Sigma)$.

\section{RLUB and its Control}
\label{sec:rlub-and-its-control}

We introduce the definition of a right linear grammar with unknown behavior (RLUB)
and its control system based on \cite{Kimoto,Kimoto2}.

Let $G=(V,\Sigma,S,P)$ be a right linear grammar.
For a subset $X$ of ${(V \cup \Sigma)}^*$ and a behavior $R$ of $G$,
we define a set $E(X,R)$ as
a set of elements in $X$ to which the first rules found in $R$ can be applied.
Formally, we define
\begin{align*}
  E(X,R) \overset{\mathrm{def}}{=} \{ x \in X \mid \exists x' \in {(V \cup \Sigma)}^* \; \exists \alpha \in \mathrm{prf}_1 (R) \mbox{ such that } x \underset{\alpha}{\Rightarrow} x' \}.
\end{align*}
For subsets $X,Y$ of ${(V \cup \Sigma)}^*$ and a behavior $R$ of $G$,
we say that {\em $Y$ is generated from $X$ by $R$ in the erasing mode}, written $X {\underset{R}{\Rightarrow}}^{e} Y$, if 
\begin{align}
  \label{eq:arrowR}
  Y = \{ y \in {(V \cup \Sigma)}^* \mid \exists x \in E(X,R) \exists \alpha \in R (x \underset{\alpha}{\Rightarrow} y) \}.
\end{align}
In order to generate $Y$ from $X$ by $R$ in the erasing mode, we apply $\alpha \in R$ to elements in $E(X,R)$.
Therefore, the elements in $X-E(X,R)$ are erased from $X$.
Note that if $\epsilon \in R$ holds, Y contains all the elements in $E(X,R)$.

Let $\mu_t$ and $\mu_b$ be non-empty intervals.
A pair $M=(\mu_t,\mu_b)$ is called a {\em generative condition} ({\em GC}, for short)
if $\mu_b \subseteq \mu_t$ holds.

A {\em right linear grammar with unknown behavior} ({\em RLUB}, for short) is a pair $H=(G,M)$, where
$G=(V,\Sigma,S,P)$ is a right linear grammar and $M=(\mu_t,\mu_b)$ is a GC.
A {\em control system for $H$} is a triple $C=(\Gamma,\phi,T)$, where
$\Gamma$ is a finite alphabet, called a {\em control alphabet}, such that $\Gamma \cap ( V \cup \Sigma ) = \emptyset$,
$\phi:\Gamma \to 2^P$ is an {\em injective} function, called a {\em control function}, and
$T$ is a subset of ${\Gamma}^{+}$, called a {\em set of control sequences}.
\begin{rem}
  In \cite{Kimoto,Kimoto2}, GCs $M=(\mu_t,\mu_b)$ with possibly infinite $\mu_t$ and $\mu_b$ are considered,
  where an infinite interval is defined as an infinite set of integers $\{ i \in \mathbb{Z}_{\geq 0} \mid i \geq k \}$ for some $k \in \mathbb{Z}_{\geq 0}$.
  However, in this paper, we focus on the generation of classes of {\em finite} languages.
  Therefore, this paper requires that $\mu_t$ and $\mu_b$ are finite.
\end{rem}
\begin{rem}
  \label{rem:phi-injection-reason}
  In \cite{Kimoto,Kimoto2}, control systems $C=(\Gamma,\phi,T)$ with $\phi$ being not injective are also considered.
  However, such control systems do not make much practical sense by the following reason.
  Consider a control system $C=(\Gamma,\phi,T)$ such that $\phi$ is not injective.
  Then, there are at least two control symbols $t_1$ and $t_2$ in $\Gamma$ such that $\phi(t_1)=\phi(t_2)$.
  From application point of view,
  we can say that this control system has a useless control symbol since the use of $t_1$ can be replaced by that of $t_2$.
  Furthermore, the existence of such extra control symbol requires extra experimental adjustment of reaction conditions.
  Therefore, control systems whose control functions are not injective do not make much practical sense.
  Thus, this paper requires that $\phi$ is injective.
\end{rem}
Let $\tau = t_1 t_2 \cdots t_n$ be a control sequence for some $n$ in ${\mathbb{Z}_{\geq 1}}$ and some $t_i$'s in $\Gamma$.
We say that {\em $H$ and $C$ generate a language $L$ using $\tau$ in the erasing mode}, written $L_e(H,C,\tau)=L$, if
for any behaviors $R_1 , R_2 , \ldots , R_n$ such that ${\phi(t_i)}^{\mu_b} \subseteq R_i \subseteq {\phi(t_i)}^{\mu_t} (\mbox{for } i=1,2,\ldots,n)$
\footnote{In the representation ${\phi(t)}^{\mu}$, we regard $\phi(t)$ as a set of symbols.
Therefore, ${\phi(t)}^{\mu} = \{ \alpha \in {\phi(t)}^* \mid |\alpha| \in \mu \}$ holds.},
there exist $X_1 , X_2 , \ldots, X_n$ such that the following two conditions are satisfied:
\begin{enumerate}
  \renewcommand{\labelenumi}{(r\arabic{enumi})}
\item $\{ S \}
  {\underset{R_1}{\Rightarrow}}^{e} X_1
  {\underset{R_2}{\Rightarrow}}^{e} \cdots
  {\underset{R_n}{\Rightarrow}}^{e} X_n$, and
\item $X_n \cap {\Sigma}^* = L$.
\end{enumerate}
The language $L_e(H,C,\tau)$ is not defined if such $L$ does not exist.
We say that {\em $H$ and $C$ synchronously generate a language $L$ using $\tau$ in the erasing mode}, written $L_e^{\mathrm{syn}}(H,C,\tau)=L$,
if the following condition (r3), as well as (r1) and (r2), is satisfied:
\begin{enumerate}
  \renewcommand{\labelenumi}{(r\arabic{enumi})}
  \setcounter{enumi}{2}
\item $(X_1 \cup X_2 \cup \cdots \cup X_{n-1}) \cap {\Sigma}^* = \emptyset$.
  \label{page:condition_r3}
\end{enumerate}
The language $L_e^{\mathrm{syn}}(H,C,\tau)$ is not defined if such $L$ does not exist.
We say that {\em $H$ and $C$ generate a class $\mathcal{L}$ of languages in the erasing mode}, written $\mathcal{L}_e(H,C)=\mathcal{L}$,
if the following two conditions are satisfied:
\begin{enumerate}
  \renewcommand{\labelenumi}{(R\arabic{enumi})}
\item for any $\tau \in T$,  $L_e(H,C,\tau)$ is defined, and
\item $\mathcal{L} = \{ L_e(H,C,\tau) \mid \tau \in T \}$.
\end{enumerate}
The class $\mathcal{L}_e(H,C)$ of languages is not defined if the above condition (R1) is not satisfied.
We say that {\em $H$ and $C$ synchronously generate a class $\mathcal{L}$ of languages in the erasing mode}, written $\mathcal{L}_e^{\mathrm{syn}}(H,C)=\mathcal{L}$,
if the following two conditions are satisfied:
\begin{enumerate}
  \renewcommand{\labelenumi}{(R\arabic{enumi}')}
\item for any $\tau \in T$,  $L_e^{\mathrm{syn}}(H,C,\tau)$ is defined, and
\item $\mathcal{L} = \{ L_e^{\mathrm{syn}}(H,C,\tau) \mid \tau \in T \}$.
\end{enumerate}
The class $\mathcal{L}_e^{\mathrm{syn}}(H,C)$ of languages is not defined if the above condition (R1') is not satisfied.

\begin{rei}
  \label{rei:generation}

  Let $H=(G,M)$ be an RLUB,
  where $G = (V,\Sigma,S,P)$,
  $V=\{ S,A,B,C \}$,
  $\Sigma= \{ a,b,c,d \}$,
  $P= \{ r_1 : S \to aA , \; r_2 : A \to bB , \; r_3 : A \to b , \; r_4 : B \to \epsilon, \; r_5 : B \to cC , \; r_6 : A \to c , \; r_7 : B \to d \}$,
  $M = (\mu_t, \mu_b)$,
  $\mu_t = [1,2]$,
  $\mu_b = [1,1]$.
  Let $C=(\Gamma,\phi,T)$ be a control system for $H$,
  where
  $\Gamma = \{ t_1, t_2, t_3, t_4 \}$,
  $\phi(t_1) = \{ r_1,r_2 \}$,
  $\phi(t_2) = \{ r_3,r_4,r_5 \}$,
  $\phi(t_3) = \{ r_3,r_6 \}$,
  $\phi(t_4) = \{ r_3,r_6,r_7 \}$,
  $T= \{ t_1 t_2, t_1 t_3 \}$.
  We have $L_e^{\mathrm{syn}}(H,C,t_1 t_2) = \{ ab \}$
  since
  we have that for any behaviors $R_1$ and $R_2$ such that\footnote{
  Note that we have
  ${\phi(t_1)}^{\mu_b} = \{ r_1, r_2 \}$,
  ${\phi(t_1)}^{\mu_t} = \{ r_1,r_2,r_1r_1,r_1r_2,r_2r_1,r_2r_2 \}$,
  ${\phi(t_2)}^{\mu_b} = \{ r_3,r_4,r_5 \}$, and
  ${\phi(t_2)}^{\mu_t} = \{ r_3,r_4,r_5,r_3r_3,r_3r_4,r_3r_5,r_4r_3,r_4r_4,r_4r_5,r_5r_3,r_5r_4,r_5r_5 \}$.
  }
  \begin{align}
    {\phi(t_1)}^{\mu_b} \subseteq & R_1 \subseteq {\phi(t_1)}^{\mu_t} \mbox{ and} \label{eq:rei:generation:r1} \\
    {\phi(t_2)}^{\mu_b} \subseteq & R_2 \subseteq {\phi(t_2)}^{\mu_t}, \label{eq:rei:generation:r2}
  \end{align}
  there exist $X_1$ and $X_2$ such that the following three conditions are satisfied:
  \begin{enumerate}
    \renewcommand{\labelenumi}{(r\arabic{enumi})}
  \item $\{ S \} {\underset{R_1}{\Rightarrow}}^{e} X_1 {\underset{R_2}{\Rightarrow}}^{e} X_2$,
  \item $X_2 \cap {\Sigma}^* = \{ ab \}$, and
  \item $X_1 \cap {\Sigma}^* = \emptyset$.
  \end{enumerate}
  For example, for $R_1 = \{ r_1,r_2 \}$ and $R_2 = \{ r_3,r_4,r_5 \}$,
  the sets $X_1 = \{ aA \}$ and $X_2 = \{ ab \}$ satisfy the above three conditions (r1),(r2), and (r3).
  For example, for $R_1 = \{ r_1,r_2,r_1r_2 \}$ and $R_2 = \{ r_3,r_4,r_5,r_3r_3,r_3r_4,r_5r_5 \}$,
  the sets $X_1 = \{ aA,abB \}$ and $X_2 = \{ ab,abcC \}$ satisfy the above three conditions (r1),(r2), and (r3).
  In this way, we can verify that for any behaviors $R_1$ and $R_2$ satisfying (\ref{eq:rei:generation:r1}) and (\ref{eq:rei:generation:r2}), there exist $X_1$ and $X_2$ satisfying the above three conditions (r1),(r2), and (r3).
  However, there are $2^4 \times 2^9$ combinations of $R_1$ and $R_2$ satisfying (\ref{eq:rei:generation:r1}) and (\ref{eq:rei:generation:r2}).
  Therefore, it is hard to verify that there exist $X_1$ and $X_2$ for any such pair of $R_1$ and $R_2$.
  Actually, by Theorem \ref{theo:generation-kimoto} mentioned later,
  it suffices to verify the only two cases:
  the case of $R_1 = {\phi(t_1)}^{\mu_b}$ and $R_2 = {\phi(t_2)}^{\mu_b}$ and the case of $R_1 = {\phi(t_1)}^{\mu_t}$ and $R_2 = {\phi(t_2)}^{\mu_t}$.
  
  We can also verify that $L_e^{\mathrm{syn}}(H,C,t_1 t_3) = \{ ab,ac \}$.
  Therefore, we have $\mathcal{L}_e^{\mathrm{syn}}(H,C)=\{ \{ ab \}, \{ ab,ac \} \}$.

  Here, we consider the control system $C'=(\Gamma,\phi,T')$ for $H$, where $T'=\{ t_1t_2,t_1t_3,t_1t_4 \}$.
  Then, we have that $L_e^{\mathrm{syn}}(H,C',t_1t_4)$ is not defined.
  We will show this by contradiction.
  Assume that $L_e^{\mathrm{syn}}(H,C',t_1t_4)$ is defined.
  Let $L$ be a language such that $L_e^{\mathrm{syn}}(H,C',t_1t_4)=L$ holds.
  We have that for any behaviors $R_1$ and $R_2$ such that
  ${\phi(t_1)}^{\mu_b} \subseteq R_1 \subseteq {\phi(t_1)}^{\mu_t}$ and
  ${\phi(t_4)}^{\mu_b} \subseteq R_2 \subseteq {\phi(t_4)}^{\mu_t}$,
  there exist $X_1$ and $X_2$ such that the following three conditions are satisfied:
  \begin{enumerate}
    \renewcommand{\labelenumi}{(r\arabic{enumi})}
  \item $\{ S \} {\underset{R_1}{\Rightarrow}}^{e} X_1 {\underset{R_2}{\Rightarrow}}^{e} X_2$,
  \item $X_2 \cap {\Sigma}^* = L$, and
  \item $X_1 \cap {\Sigma}^* = \emptyset$.
  \end{enumerate}
  For $R_1 = {\phi(t_1)}^{\mu_b}$ and $R_2 = {\phi(t_4)}^{\mu_b}$,
  we have $\{ S \} {\underset{R_1}{\Rightarrow}}^{e} \{ aA \} {\underset{R_2}{\Rightarrow}}^{e} \{ ab,ac \}$.
  Therefore, by (r2), $L=\{ ab,ac \}$ hods.
  However, since for $R_1 = {\phi(t_1)}^{\mu_t}$ and $R_2 = {\phi(t_4)}^{\mu_t}$,
  we have $\{ S \} {\underset{R_1}{\Rightarrow}}^{e} \{ aA,abB \} {\underset{R_2}{\Rightarrow}}^{e} \{ ab,ac,abd \}$,
  we have $L=\{ ab,ac,abd \}$, which contradicts $L=\{ ab,ac \}$.
  Therefore, we have that $L_e^{\mathrm{syn}}(H,C',t_1t_4)$ is not defined.
  Moreover, we have that $\mathcal{L}_e^{\mathrm{syn}}(H,C')$ is not defined.
  \halmos
\end{rei}


In \cite{Kimoto}, Kimoto et al. proved the following Theorem \ref{theo:generation-kimoto}.
\begin{theo}
  \label{theo:generation-kimoto}
  (Theorem 2 in \cite{Kimoto} for RLUBs)
  Let $M=(\mu_t,\mu_b)$ be a GC,
  $G=(V,\Sigma,S,P)$ be a right linear grammar,
  $L$ be a language over $\Sigma$, and
  $C=(\Gamma,\phi,T)$ be a control system.
  Consider RLUBs $H=(G,M)$, $H_1=(G,(\mu_t,\mu_t))$, $H_2=(G,(\mu_b,\mu_b))$, and
  any control sequence $\tau \in T$.
  The equality $L_e(H,C,\tau) = L$ holds if and only if $L_e(H_1,C,\tau) = L_e(H_2,C,\tau) = L$ holds.
  The equality $L_e^{\mathrm{syn}}(H,C,\tau) = L$ holds if and only if $L_e^{\mathrm{syn}}(H_1,C,\tau) = L_e^{\mathrm{syn}}(H_2,C,\tau) = L$ holds.
\end{theo}

We will introduce new notions and notations related to RLUBs
which were not defined in \cite{Kimoto,Kimoto2}.
Let $H=(G,M)$ be an RLUB and $C=(\Gamma,\phi,T)$ be a control system for $H$,
where $G=(V,\Sigma,S,P)$.
A binary relation $\preceq_C$ over $\Gamma$ is defined as follows:
\begin{align*}
  \preceq_C \; \overset{\mathrm{def}}{=} \; \{ (t_1,t_2) \in {\Gamma}^2 \mid \phi(t_1) \subseteq \phi(t_2) \}.
\end{align*}
For example, the control system $C$ defined in Example \ref{rei:generation} satisfies $\preceq_C = \{ (t,t) \mid t \in \Gamma \} \cup \{ (t_3,t_4) \}$.
A binary relation is a {\em partial order}
if it is reflexive, antisymmetric,
and transitive(\cite{birkhoff-lattice})
\footnote{
A binary relation $R$ over $U$ is reflexive if for any $a\in U$,
$aRa$ holds.
It is antisymmetric if for any $a,b\in U$,
$aRb$ and $bRa$ imply $a=b$.
It is transitive if for any $a,b,c\in U$,
$aRb$ and $bRc$ imply $aRc$.
}.
It is straightforward to show the following Proposition \ref{prop:preceq-quasiorder-partialorder} since $\phi$ is an injection.
\begin{prop}
  \label{prop:preceq-quasiorder-partialorder}
  Let $C=(\Gamma,\phi,T)$ be a control system.
  The binary relation $\preceq_C$ is a partial order.
\end{prop}

\section{Important Definitions and Lemma}
\label{sec:Important-Definitions-and-Lemma}

For an integer sequence $\lambda \in {\mathbb{Z}_{\geq 0}}^+$,
we define $\sum \lambda \overset{\mathrm{def}}{=} \sum_{k=1}^{|\lambda|} \lambda[k]$.
For example, for $\lambda=<5,3,4>$, we have $\sum \lambda = 5+3+4=12$.

For an integer sequence $\lambda \in {\mathbb{Z}_{\geq 0}}^+$ and an integer $i$ such that $1 \leq i \leq \sum \lambda$,
we define $(\lambda,i)$ as the smallest integer $j$ such that $\sum_{k=1}^{j} \lambda[k] \geq i$.
For example, for $\lambda=<5,3,4>$, we have $(\lambda,1)=\cdots=(\lambda,5)=1$, $(\lambda,6)=\cdots=(\lambda,8)=2$, and $(\lambda,9)=\cdots=(\lambda,12)=3$.
We have the following Fact \ref{fact:lambda}.
\begin{fact}
  \label{fact:lambda}
  Let $\lambda \in {\mathbb{Z}_{\geq 0}}^+$ and $j \in [1,|\lambda|]$.
  Then, $(\lambda,i ) = j$ holds for any integer $i$ such that $\sum_{k=1}^{j-1} \lambda[k] +1 \leq i \leq \sum_{k=1}^{j} \lambda[k]$,
  where the sum of the empty integer sequences is defined as zero\footnote{
  The sum $\displaystyle \sum_{k=n}^{m} \lambda[k]$ is $0$ if $n>m$ holds.
  }.
\end{fact}

Let $\Gamma$ be a finite control alphabet.
We define $\Phi (\Gamma) \overset{\mathrm{def}}{=} \{ (\tau,\lambda) \mid \tau \in \Gamma^+, \lambda \in {\mathbb{Z}_{\geq 0}}^+, |\tau|=|\lambda| \}$.
We will define a binary relation over $\Phi (\Gamma)$, which is one of the most important definitions in this paper.
\begin{defi}
  \label{def:Rightarrow-M-preceq}
  Let $\Gamma$ be a finite control alphabet.
  Let $\tau_1, \tau_2 \in {\Gamma}^+$, and $\lambda_1,\lambda_2 \in {\mathbb{Z}_{\geq 0}}^+$
  such that $(\tau_1,\lambda_1),(\tau_2,\lambda_2) \in \Phi(\Gamma)$ holds.
  Let $M=(\mu_{t},\mu_{b})$ be a GC, and $\preceq$ be a partial order over $\Gamma$.
  We write $(\tau_1,\lambda_1) \underset{M,\preceq}{\Rightarrow} (\tau_2,\lambda_2)$ if the following three conditions hold:
  \begin{align*}
    &\mbox{(A1)} \quad \sum \lambda_1 = \sum \lambda_2 = m \mbox{ holds for some } m \in \mathbb{Z}_{\geq 0},\\
    &\mbox{(A2)} \quad \lambda_1 \in {\mu_{b}}^{|\tau_1|} \mbox{ implies } \lambda_2 \in {\mu_{t}}^{|\tau_2|}, \mbox{ and} \\
    &\mbox{(A3)} \quad \mbox{for any } i \in [1,m] , \; \tau_1[(\lambda_1,i)] \preceq \tau_2[(\lambda_2,i)] \mbox{ holds}.
  \end{align*}
  \halmos
\end{defi}
It is straightforward to show that the binary relation $\underset{M,\preceq}{\Rightarrow}$ is reflexive.
Note that the binary relation $\underset{M,\preceq}{\Rightarrow}$ is not transitive.
For example, for a GC $M=([3,11],[4,8])$ and a control symbol $t$,
we have $(t,<6>) \underset{M,\preceq}{\Rightarrow} (tt,<3,3>)$ and
$(tt,<3,3>) \underset{M,\preceq}{\Rightarrow} (tt,<1,5>)$,
but we have that $(t,<6>) \underset{M,\preceq}{\Rightarrow} (tt,<1,5>)$ does not hold since (A2) of Definition \ref{def:Rightarrow-M-preceq} does not hold.

\begin{rei}
  \label{rei:Rightarrow}
  Let $M=([3,11],[4,8])$ be a GC,
  $\Gamma=\{ t_1,t_2,t_3 \}$ be a finite control alphabet,
  and $\preceq = \{ (t,t) \mid t \in \Gamma \} \cup \{ (t_2,t_3) \}$ be a partial order over $\Gamma$.
  \begin{enumerate}
    \renewcommand{\labelenumi}{(\alph{enumi})}
  \item
    We have $(t_1 t_2 t_3,<4,5,6>) \underset{M,\preceq}{\Rightarrow} (t_1 t_3,<4,11>)$ because the following three conditions hold:
    \begin{align*}
      & \mbox{(A1)}\quad \sum <4,5,6> = \sum <4,11> = 15 \mbox{ holds}, \displaybreak[1] \\
      & \mbox{(A2)}\quad <4,5,6> \in {[4,8]}^{|t_1 t_2 t_3|} \mbox{ and } <4,11> \in {[3,11]}^{|t_1 t_3|} \mbox{ hold, and} \displaybreak[1] \\
      & \mbox{(A3)}\quad \mbox{for any } i \in [1,15] , \; t_1 t_2 t_3[(<4,5,6>,i)] \preceq t_1 t_3[(<4,11>,i)] \mbox{ holds},
    \end{align*}
    where the third condition can be verified by the following expressions:
    \begin{alignat*}{5}
      & t_1 t_2 t_3[(<4,5,6>,1)] = t_1 t_2 t_3[1]=t_1 \quad& &\preceq \quad&  & t_1 = t_1 t_3 [1] = t_1 t_3[(<4,11>,1)], \displaybreak[1] \\
      & t_1 t_2 t_3[(<4,5,6>,2)] = t_1 t_2 t_3[1]=t_1 & &\preceq& &             t_1 = t_1 t_3 [1] = t_1 t_3[(<4,11>,2)], \displaybreak[1] \\
      & t_1 t_2 t_3[(<4,5,6>,3)] = t_1 t_2 t_3[1]=t_1 & &\preceq& &             t_1 = t_1 t_3 [1] = t_1 t_3[(<4,11>,3)], \displaybreak[1] \\
      & t_1 t_2 t_3[(<4,5,6>,4)] = t_1 t_2 t_3[1]=t_1 & &\preceq& &             t_1 = t_1 t_3 [1] = t_1 t_3[(<4,11>,4)], \displaybreak[1] \\
      & t_1 t_2 t_3[(<4,5,6>,5)] = t_1 t_2 t_3[2]=t_2 & &\preceq& &             t_3 = t_1 t_3 [2] = t_1 t_3[(<4,11>,5)], \displaybreak[1] \\
      & & & \vdots \displaybreak[1] \\
      & t_1 t_2 t_3[(<4,5,6>,15)] = t_1 t_2 t_3[3]=t_3 & &\preceq& &             t_3 = t_1 t_3 [2] = t_1 t_3[(<4,11>,15)].
    \end{alignat*}
  \item
    We have that there exists no integer sequence $\lambda' \in {\mathbb{Z}_{\geq 1}}^{+}$ such that $(t_1 t_1 t_3,<4,5,6>) \underset{M,\preceq}{\Rightarrow} (t_1 t_2, \lambda')$.
    Assume that there exists an integer sequence $\lambda' \in {\mathbb{Z}_{\geq 1}}^{+}$ such that $(t_1 t_1 t_3,<4,5,6>) \underset{M,\preceq}{\Rightarrow} (t_1 t_2, \lambda')$.
    We can write $\lambda'=<x_1,x_2>$ since $|\lambda'|=|t_1 t_2|=2$.
    By (A1) of Definition \ref{def:Rightarrow-M-preceq}, $x_1 + x_2 = \sum <4,5,6> = 15$ holds.
    By (A2) of Definition \ref{def:Rightarrow-M-preceq}, we have $x_2 \in [3,11]$, which implies $x_2 \geq 3 \geq 1$.
    Therefore, $(<x_1,x_2>,15)=2$ holds, which implies $t_1 t_2 [(\lambda',15)]=t_1 t_2[2]=t_2$.
    However, we have $t_1 t_1 t_3[(<4,5,6>,15)]=t_1 t_1 t_3[3]=t_3$, which contradicts (A3) of Definition \ref{def:Rightarrow-M-preceq}.
    \halmos
  \end{enumerate}
\end{rei}
By generalizing (b) of Example \ref{rei:Rightarrow}, we obtain the following Proposition \ref{prop:haji-ga-not-preceq}.
\begin{prop}
  \label{prop:haji-ga-not-preceq}
  Let $\Gamma$ be a finite control alphabet,
  $t_1, t_2 \in \Gamma$ be control symbols, and
  $\tau_1, \tau_2 \in \Gamma^*$ be control sequences.
  Let $\lambda_1, \lambda_2 \in {\mathbb{Z}_{\geq 1}}^+$ be integer sequences.
  Let $M$ be a GC and
  $\preceq$ be a partial order over $\Gamma$.
  If $t_1 \preceq t_2$ does not hold,
  then
  neither $(t_1 \tau_1,\lambda_1) \underset{M,\preceq}{\Rightarrow} (t_2 \tau_2, \lambda_2)$ nor
  $(\tau_1 t_1,\lambda_1) \underset{M,\preceq}{\Rightarrow} (\tau_2 t_2, \lambda_2)$ holds.
  
\end{prop}


\begin{proof}

  Since $\lambda_1[1] \geq 1$ and $\lambda_2[1] \geq 1$ hold,
  we have $(\lambda_1,1)=(\lambda_2,1)=1$ by Fact \ref{fact:lambda},
  which implies $t_1 \tau_1[(\lambda_1,1)]=t_1 \tau_1 [1]=t_1$ and $t_2 \tau_2[(\lambda_2,1)]=t_2 \tau_2 [1]=t_2$.
  Thus, the condition (A3) of Definition \ref{def:Rightarrow-M-preceq} does not hold,
  which implies that $(t_1 \tau_1,\lambda_1) \underset{M,\preceq}{\Rightarrow} (t_2 \tau_2, \lambda_2)$ does not hold.

  Assume that $(\tau_1 t_1,\lambda_1) \underset{M,\preceq}{\Rightarrow} (\tau_2 t_2, \lambda_2)$ holds.
  Then, we have that $\sum \lambda_1 = \sum \lambda_2 = m$ holds for some $m \in \mathbb{Z}_{\geq 1}$.
  Since $\lambda_1[|\lambda_1|] \geq 1$ and $\lambda_2[|\lambda_2|] \geq 1$ hold,
  we have $(\lambda_1,m)=|\lambda_1|$ and $(\lambda_2,m)=|\lambda_2|$ by Fact \ref{fact:lambda},
  which implies $\tau_1 t_1[(\lambda_1,m)]=\tau_1 t_1[|\lambda_1|]=t_1$ and $\tau_2 t_2[(\lambda_2,m)]=\tau_2 t_2[|\lambda_2|]=t_2$.
  Thus, the condition (A3) of Definition \ref{def:Rightarrow-M-preceq} does not hold, which is a contradiction.
\end{proof}

We have the following Proposition \ref{prop:preceq-subset-rightarrow-subset}.
\begin{prop}
  \label{prop:preceq-subset-rightarrow-subset}

  Let $M$ be a GC.
  Let $\Gamma$ be a finite control alphabet.
  Let $\preceq_1$ and $\preceq_2$ be partial orders over $\Gamma$.
  Assume that $\preceq_1 \; \subseteq \; \preceq_2$ holds.
  Then, we have $\underset{M,\preceq_1}{\Rightarrow} \; \subseteq \; \underset{M,\preceq_2}{\Rightarrow}$.
  
\end{prop}


\begin{proof}

  Let $(\tau_1,\lambda_1),(\tau_2,\lambda_2) \in \Phi(\Gamma)$.
  Assume that $(\tau_1,\lambda_1) \underset{M,\preceq_1}{\Rightarrow} (\tau_2,\lambda_2)$ holds.
  Then, by Definition \ref{def:Rightarrow-M-preceq}, we have that
  \begin{align*}
    &\mbox{(A1)} \quad \sum \lambda_1 = \sum \lambda_2 = m \mbox{ holds for some } m \in \mathbb{Z}_{\geq 0},\\
    &\mbox{(A2)} \quad \lambda_1 \in {\mu_{b}}^{|\tau_1|} \mbox{ implies } \lambda_2 \in {\mu_{t}}^{|\tau_2|}, \mbox{ and} \\
    &\mbox{(A3)} \quad \mbox{for any } i \in [1,m] , \; \tau_1[(\lambda_1,i)] \preceq_1 \tau_2[(\lambda_2,i)] \mbox{ holds}.
  \end{align*}
  Since $\preceq_1 \; \subseteq \; \preceq_2$ holds, by (A3), we have that
  \begin{align*}
    &\mbox{(A3)'} \quad \mbox{for any } i \in [1,m] , \; \tau_1[(\lambda_1,i)] \preceq_2 \tau_2[(\lambda_2,i)] \mbox{ holds}.
  \end{align*}
  Therefore, by (A1), (A2), and (A3)', we have $(\tau_1,\lambda_1) \underset{M,\preceq_2}{\Rightarrow} (\tau_2,\lambda_2)$.
  Thus, we have $\underset{M,\preceq_1}{\Rightarrow} \; \subseteq \; \underset{M,\preceq_2}{\Rightarrow}$.
\end{proof}

In section \ref{sec:rlub-and-its-control},
we defined a binary relation ${\underset{R}{\Rightarrow}}^{e}$.
Here, by giving a restriction to ${\underset{R}{\Rightarrow}}^{e}$,
we will define a binary relation $\overset{l}{\underset{R}{\Rightarrow^{e}}}$ for $l \in \mathbb{Z}_{\geq 0}$.
\begin{defi}
  \label{defi:rightarrow-R-l}
  Let $G=(V,\Sigma,S,P)$ be a right linear grammar.
  For subsets $X,Y$ of ${(V \cup \Sigma)}^*$,
  a behavior $R$ of $G$, and
  $l \in \mathbb{Z}_{\geq 0}$,
  we say that {\em $Y$ is generated from $X$ by $R$ in $l$ steps in the erasing mode}, written $X \overset{l}{\underset{R}{\Rightarrow^{e}}} Y$, if
  \begin{align}
    \label{eq:arrowRl}
    Y = \{ y \in {(V \cup \Sigma)}^* \mid \exists x \in E(X,R) \exists \alpha \in R (|\alpha|=l \land x \underset{\alpha}{\Rightarrow} y) \}.
  \end{align}
  \halmos
\end{defi}
Note that the following Remark \ref{rem:subset-rightarrow} and Remark \ref{rem:l-in-mu} hold.
\begin{rem}
  \label{rem:subset-rightarrow}
  If we have $X \overset{l}{\underset{R}{\Rightarrow^{e}}} Y$ and $X \underset{R}{\Rightarrow^e} Y'$,
  the set $Y$ is a subset of $Y'$
  since the difference between (\ref{eq:arrowR}) and (\ref{eq:arrowRl}) is just the expression $|\alpha|=l$.
\end{rem}
\begin{rem}
  \label{rem:l-in-mu}
  Let $\mu$ be an interval and $t \in \Gamma$.
  Assume that $X \overset{l}{\underset{{\phi(t)}^{\mu}}{\Rightarrow^{e}}} Y$ holds.
  Then, we have that $l \not\in \mu$ implies $Y = \emptyset$.
\end{rem}
The relation $X \overset{l}{\underset{R}{\Rightarrow^{e}}} Y$ can be extendedly defined for the case that $R$ and $l$ are sequences in the following way.
\begin{defi}
  \label{defi:rightarrow-gamma-lambda}
  Let $\gamma=<R_1,\ldots,R_n>$ be a sequence of behaviors for some $n \in \mathbb{Z}_{\geq 1}$ and some $R_i$'s, where $R_i$'s are bahaviors.
  Let $\lambda=<l_1,\ldots,l_n>$ be a integer sequence for some $l_i$'s in $\mathbb{Z}_{\geq 0}$.
  We write $X \overset{\lambda}{\underset{\gamma}{\Rightarrow^e}} Y$
  if there exist $X_0,\ldots,X_n \subseteq {(V \cup \Sigma)}^*$ such that
  \begin{align*}
    X = X_0 \overset{l_1}{\underset{R_1}{\Rightarrow^{e}}} X_1 \overset{l_2}{\underset{R_2}{\Rightarrow^{e}}} \cdots \overset{l_n}{\underset{R_n}{\Rightarrow^{e}}} X_n = Y.
  \end{align*}
  \halmos
\end{defi}

\begin{rei}
  Let $G=(V,\Sigma,S,P)$ be a right linear grammar,
  where $V=\{ S \}$, $\Sigma=\{ a,b,c \}$, and $P=\{r_1:S\to aS, \; r_2:S\to bS, \; r_3:S\to c \}$.
  Let $R_1 = \{ {r_1}^2, {r_2}^3 \}$ and $R_2=\{ {r_1}^2,r_1r_3, {r_2}^3 \}$ be behaviors of $G$.
  Then, we have $\{ S \} \overset{3}{\underset{R_1}{\Rightarrow^{e}}} \{ b^3 S \} \overset{2}{\underset{R_2}{\Rightarrow^{e}}} \{ b^3 a^2 S, b^3 a c \}$.
  Therefore, we have $ \{ S \} \overset{<3,2>}{\underset{<R_1,R_2>}{\Rightarrow^e}} \{ b^3 a^2 S, b^3 a c \}$.
  \halmos
\end{rei}

Let $C=(\Gamma,\phi,T)$ be a control system.
For a control sequence $\tau=t_1 \cdots t_m \; (t_i \in \Gamma \mbox{ for } i \in [1,m])$ and
an interval $\mu$,
we often consider a sequence $< {\phi(t_1)}^{\mu},\ldots,{\phi(t_m)}^{\mu} >$ of behaviors.
Therefore, we introduce the following Definition \ref{defi:phi-tau-mu}.
\begin{defi}
  \label{defi:phi-tau-mu}
  Let $C=(\Gamma,\phi,T)$ be a control system.
  For a control sequence $\tau =t_1 \cdots t_m \; (t_i \in \Gamma \mbox{ for } i \in [1,m])$ and
  an interval $\mu$,
  we define ${\phi(\tau)}^{\mu} \overset{\mathrm{def}}{=} < {\phi(t_1)}^{\mu},\ldots,{\phi(t_m)}^{\mu} >$.
  \halmos
\end{defi}

We have the following Lemma \ref{lemm:binaryrelation-subset}.
\begin{lemm}
  \label{lemm:binaryrelation-subset}
  Let $H=((V,\Sigma,S,P),({\mu}_{t},{\mu}_{b}))$ be an RLUB and
  $C=(\Gamma,\phi,T)$ be a control system for $H$.
  Let $W,X,Y,Z \subseteq {(V \cup \Sigma)}^*$
  such that $W \subseteq Y$.
  Let $\alpha,\beta \in \Gamma^+$,
  $\lambda_\alpha \in {\mathbb{Z}_{\geq 0}}^+$, and
  $\lambda_\beta \in {\mathbb{Z}_{\geq 1}}^+$
  such that $(\alpha,\lambda_\alpha) \underset{M,\preceq_C}{\Rightarrow} (\beta,\lambda_\beta)$,
  $W \overset{\lambda_\alpha}{\underset{{\phi(\alpha)}^{{\mu}_{b}}}{\Rightarrow^e}} X$, and
  $Y \overset{\lambda_\beta}{\underset{{\phi(\beta)}^{{\mu}_{t}}}{\Rightarrow^e}} Z$ hold.
  Then, $X \subseteq Z$ holds.
\end{lemm}


\begin{proof}
  
  Let $X_i \subseteq {(V \cup \Sigma)}^*$ for each $i \in [1,|\alpha|]$ such that
  \begin{align}
    \label{eq:derivation-detail-wx}
    W \overset{\lambda_\alpha[1]}{\underset{{\phi(\alpha[1])}^{{\mu}_{b}}}{\Rightarrow^e}} X_1 \overset{\lambda_\alpha[2]}{\underset{{\phi(\alpha[2])}^{{\mu}_{b}}}{\Rightarrow^e}} \cdots \overset{\lambda_\alpha[|\lambda_\alpha|]}{\underset{{\phi(\alpha[|\alpha|])}^{{\mu}_{b}}}{\Rightarrow^e}} X_{|\alpha|}=X.
  \end{align}
  Let $Y_i \subseteq {(V \cup \Sigma)}^*$ for each $i \in [1,|\beta|]$ such that
  \begin{align}
    \label{eq:derivation-detail-yz}
    Y \overset{\lambda_\beta[1]}{\underset{{\phi(\beta[1])}^{{\mu}_{t}}}{\Rightarrow^e}} Z_1 \overset{\lambda_\beta[2]}{\underset{{\phi(\beta[2])}^{{\mu}_{t}}}{\Rightarrow^e}} \cdots \overset{\lambda_\beta[|\lambda_\beta|]}{\underset{{\phi(\beta[|\beta|])}^{{\mu}_{t}}}{\Rightarrow^e}} Z_{|\beta|}=Z.
  \end{align}
  Note that $|\alpha|=|\lambda_\alpha|$ and $|\beta|=|\lambda_\beta|$ hold
  since the relation $\underset{M,\preceq_C}{\Rightarrow}$ is defined over $\Phi(\Gamma)$.
  
  By $(\alpha,\lambda_\alpha) \underset{M,\preceq_C}{\Rightarrow} (\beta,\lambda_\beta)$, we have that
  \begin{align}
    & \sum \lambda_{\alpha} = \sum \lambda_{\beta} = n \mbox{ holds for some } n \in \mathbb{Z}_{\geq 0},\label{eq:rightarrow-def1} \\
    & \lambda_{\alpha} \in {\mu_{b}}^{|\alpha|} \mbox{ implies } \lambda_{\beta} \in {\mu_{t}}^{|\beta|}, \mbox{ and} \label{eq:rightarrow-def2} \\
    & \mbox{for any } i \in [1,n] , \alpha[(\lambda_{\alpha},i)] \preceq_C \beta[(\lambda_{\beta},i)] \mbox{ holds}. \label{eq:rightarrow-def3}
  \end{align}

  If $\lambda_{\alpha} \not\in {\mu_{b}}^{|\alpha|}$ holds,
  by (\ref{eq:derivation-detail-wx}) and Remark \ref{rem:l-in-mu},
  we have $X=\emptyset$, which implies
  $X \subseteq Z$.
  Therefore, it suffices to consider the case that $\lambda_{\alpha} \in {\mu_{b}}^{|\alpha|}$ holds.
  
  By (\ref{eq:rightarrow-def2}), we have $\lambda_{\beta} \in {\mu_{t}}^{|\beta|}$.
  Assume $x \in X$.
  There exists at least one derivation process $\pi$ of $x$ which can contribute to the generation of $x \in X$ in the process (\ref{eq:derivation-detail-wx}).
  Let us write such $\pi$ as $x_0 \underset{r_1}{\Rightarrow} x_1 \underset{r_2}{\Rightarrow} \cdots \underset{r_n}{\Rightarrow} x_n$,
  where $x_i \in {(V \cup \Sigma)}^* \; (i=0,\ldots,n)$, $x_0 \in W$, $x_n=x$, and $r_i \in P \; (i=1,\ldots,n)$.
  Then, by $x \in X$ and (\ref{eq:derivation-detail-wx}), we have $r_i \in \phi(\alpha[j])$
  for any integer $j$ with $1 \leq j \leq |\lambda_\alpha|$ and
  for any integer $i$ with $\sum_{k=1}^{j-1} \lambda_\alpha [k]+1 \leq i \leq \sum_{k=1}^{j} \lambda_\alpha [k]$.
  Therefore, by Fact \ref{fact:lambda}, we have $r_i \in \phi(\alpha[(\lambda_\alpha,i)])$
  for any integer $i$ with $1 \leq i \leq n$.
  In addition, by (\ref{eq:rightarrow-def3}), we have $\phi(\alpha[(\lambda_{\alpha},i)]) \subseteq \phi(\beta[(\lambda_{\beta},i)])$ for any integer $i$ with $1 \leq i \leq n$.
  Therefore, $r_i \in \phi(\beta[(\lambda_{\beta},i)])$ holds for any integer $i$ with $1 \leq i \leq n$.
  Then, by Fact \ref{fact:lambda}, we have $r_i \in \phi(\beta[j])$
  for any integer $j$ with $1 \leq j \leq |\lambda_\beta|$ and
  for any integer $i$ with $\sum_{k=1}^{j-1} \lambda_\beta [k]+1 \leq i \leq \sum_{k=1}^{j} \lambda_\beta [k]$,
  where we should note that by $\lambda_\beta \in {\mathbb{Z}_{\geq 1}}^+$, we have $\sum_{k=1}^{j-1} \lambda_\beta [k]+1 \leq \sum_{k=1}^{j} \lambda_\beta [k]$.
  Therefore, by $\lambda_{\beta} \in {\mu_{t}}^{|\beta|}$, we have
  \begin{align*}
    r_{\sum_{k=1}^{j-1} \lambda_\beta [k]+1} \; \cdots \; r_{\sum_{k=1}^{j} \lambda_\beta [k]} \in {\phi(\beta[j])}^{\mu_t} \mbox{ for any integer } j \mbox{ with } 1 \leq j \leq |\lambda_\beta|,
  \end{align*}
  where $r_{\sum_{k=1}^{j-1} \lambda_\beta [k]+1} \; \cdots \; r_{\sum_{k=1}^{j} \lambda_\beta [k]}$ is not an empty sequence.
  In addition, $x_0 \in Y$ holds since $W \subseteq Y$ holds.
  Therefore, we have $x_{\sum_{k=1}^{j} \lambda_\beta [k]} \in Z_j$ for any integer $j$ with $1 \leq j \leq |\lambda_\beta|$.
  Then, we have $x = x_n = x_{\sum \lambda_{\beta}} \in Z_{|\lambda_\beta|} = Z_{|\beta|} = Z$.
\end{proof}

For an integer $n \in \mathbb{Z}_{\geq 1}$,
an integer $m \in \mathbb{Z}_{\geq 0}$, and
an interval $\mu$,
we define a set $\mathrm{Div}(n,m,\mu)$ as a set of integer sequences $\lambda$ such that
$|\lambda|=n$ holds,
$\Sigma \lambda = m$ holds, and
every element of $\lambda$ is in $\mu$.
Formally, we define
\begin{align*}
  \mathrm{Div}(n,m,\mu) \overset{\mathrm{def}}{=} \{ \lambda \in \mu^n \mid \sum \lambda = m \}.
\end{align*}
For example, we have $\mathrm{Div}(2,10,[4,8])=\{ <4,6>,<5,5>,<6,4> \}$, $\mathrm{Div}(3,13,[4,8])=\{ <4,4,5>,<4,5,4>,<5,4,4> \}$, and $\mathrm{Div}(3,11,[4,8])=\emptyset$.

\section{Characterization of Controlled Generation of RLUBs}
\label{sec:Characterization-of-Controlled-Generation-of-RLUBs}

In section \ref{subsec:condition-c1},
we introduce the condition (C),
which plays a very important role in this section.
Section \ref{subsec:construction-of-h-c-syn} and \ref{subsec:proof-characterization-syn} will show
that the condition (C) is necessary and sufficient for the controlled synchronous generation of RLUBs in the erasing mode.
More precisely, section \ref{subsec:construction-of-h-c-syn} gives
a method of constructing an RLUB $H_*$ and its control system $C_*$, under the assumption that the condition (C) holds.
Then, section \ref{subsec:proof-characterization-syn} shows that $H_*$ and $C_*$ synchronously generate a given language class $\mathcal{L}$.
Moreover, we show that the condition (C) is also necessary for the controlled synchronous generation of RLUBs in the erasing mode,
which leads to the characterization of the controlled synchronous generation of RLUBs in the erasing mode.
Finally, we introduce the condition (C') which is obtained by modifying (C), and
use it to characterize the controlled (possibly) non-synchronous generation of RLUBs in the erasing mode.



\subsection{Condition (C)}
\label{subsec:condition-c1}

\begin{defi}
  \label{defi:condition-syn}
  Let $\mathcal{L}$ be a finite class of non-empty finite languages over a finite alphabet $\Sigma$,
  $M=(\mu_t,\mu_b)$ be a GC, and
  $(\Gamma,\preceq)$ be a partially ordered finite control alphabet.
  Let $\theta$ be an injection from $\mathcal{L}$ to $\Gamma^+$ and
  $\delta_{\mathcal{L}}=\{ \delta_L \mid L \in \mathcal{L} \}$ be a class of Boolean functions $\delta_L$ from $L \; (\in \mathcal{L})$ to $\{ 0,1 \}$.
  We say that {\em $\theta$ and $\delta_{\mathcal{L}}$ satisfy the condition (C) with respect to $\mathcal{L}$, $M$, and $(\Gamma,\preceq)$} if the following (c1) and (c2) hold:
  \begin{align*}
    \mbox{(c1)} \; & \mathrm{Alph}(\theta(\mathcal{L})) = \Gamma \mbox{ holds}, \mbox{ and}\\
    \mbox{(c2)} \; & \mbox{for any } L \in \mathcal{L}, \mbox{ for any } w \in L, \mbox{ there exists } \lambda \in \mathrm{Div}(|\theta(L)|,|w|+\delta_{L}(w),\mu_b) \\
    & \mbox{such that the following (s1) and (s2) hold}: \\
    & \mbox{(s1)} \; \forall L' \in \mathcal{L} \left( \left( \exists \lambda' \in {\mathbb{Z}_{\geq 1}}^+ \left( (\theta(L),\lambda) \underset{M,\preceq}{\Rightarrow} (\theta(L'),\lambda') \right) \right) \mbox{ implies } w \in L' \right), \\
    & \mbox{(s2)} \; \forall L' \in \mathcal{L}, \forall \tau \in \Gamma^+
    \left(
    \begin{aligned}
      & \left( \exists \lambda' \in {\mathbb{Z}_{\geq 1}}^+ \left( (\theta(L),\lambda) \underset{M,\preceq}{\Rightarrow} (\tau,\lambda') \right) \right) \mbox{ implies } \\
      & \quad \left( \tau \mbox{ is not a proper prefix of } \theta(L') \right)
    \end{aligned}
    \right).
  \end{align*}
  \halmos
\end{defi}

We will show in Theorem \ref{theo:sufficient-syn} that the existence of $\theta$ and $\delta_{\mathcal{L}}$ satisfying (C) with respect to $\mathcal{L}$, $M$, and $(\Gamma,\preceq)$
allows us to construct an RLUB $H=(G,M)$ and a control system $C=(\Gamma,\phi,T)$ for $H$
such that
$\preceq = \preceq_C$ holds and
$H$ and $C$ synchronously generate $\mathcal{L}$ in the erasing mode.
Before showing Theorem \ref{theo:sufficient-syn},
we give examples and we give a proposition.

\begin{rei}
  \label{rei:jouken-1}
  Let $L_1=\{ a^{15} \}$, $L_2=\{ a^{15},b^{7} \}$, $L_3=\{ c^{5} \}$, and $L_4=\{ c^{5},d^{4} \}$.
  Let $\mathcal{L}_{\dag}=\{ L_1,L_2,L_3,L_4 \}$ be a finite class of non-empty finite languages.
  Let $M_{\dag}=([3,11],[4,8])$ be a GC,
  ${\Gamma_{\dag}}=\{ t_1,t_2,t_3 \}$ be a finite control alphabet,
  and $\preceq_{\dag} = \{ (t,t) \mid t \in {\Gamma_{\dag}} \} \cup \{ (t_2,t_3) \}$ be a partial order over ${\Gamma_{\dag}}$.
  Then, we define an injection $\theta_{\dag}:\mathcal{L}_{\dag} \to {\Gamma_{\dag}}^+$ as follows:
  \begin{align*}
    \theta_{\dag}(L_1) = t_1 t_2 t_3, \; \theta_{\dag}(L_2)=t_1 t_3, \; \theta_{\dag}(L_3)=t_2, \; \theta_{\dag}(L_4)=t_3.
  \end{align*}
  We define a class $\delta_{\mathcal{L}_{\dag}} = \{ \delta_{L_1},\ldots,\delta_{L_4} \}$ of Boolean functions $\delta_L$ from $L \; (\in \mathcal{L}_{\dag})$ to $\{0,1\}$ as follows:
  \begin{align*}
    \delta_{L_1}(a^{15})=0, \; \delta_{L_2}(a^{15})=0, \; \delta_{L_2}(b^{7})=1, \; \delta_{L_3}(c^{5})=0, \; \delta_{L_4}(c^{5})=0, \; \delta_{L_4}(d^{4})=0.
  \end{align*}
  Note that only $\delta_{L_2}(b^7)$ is $1$.
  We can verify that this $\theta_{\dag}$ and $\delta_{\mathcal{L}_{\dag}}$ satisfy the condition (C) with respect to $\mathcal{L}_{\dag}$, $M_{\dag}$, and $({\Gamma_{\dag}},\preceq_{\dag})$.
  Note that the condition (c1) of Definition \ref{defi:condition-syn} holds
  since $\mathrm{Alph}(\theta_{\dag}(\mathcal{L}_{\dag})) = {\Gamma_{\dag}}$ holds.
  It suffices to show that the condition (c2) of Definition \ref{defi:condition-syn} holds.

  For $L_1 \in \mathcal{L}_{\dag}$ and $a^{15} \in L_1$, we can verify that an integer sequence $\lambda_{L_1,a^{15}} = <4,5,6> \in \mathrm{Div}(|\theta_{\dag}(L_1)|,|a^{15}|+\delta_{L_1}(a^{15}),[4,8])=\mathrm{Div}(3,15,[4,8])$ satisfies the statements (s1) and (s2) as follows.
  Firstly, we verify the statement (s1).
  In the case of $L'=L_1$ or $L_2$, the statement (s1) holds since $a^{15} \in L'$.
  In the case of $L'=L_3$ or $L_4$, the statement (s1) holds since we have that by Proposition \ref{prop:haji-ga-not-preceq} there exists no integer sequence $\lambda' \in {\mathbb{Z}_{\geq 1}}^+$ such that $(\theta_{\dag}(L_1),<4,5,6>) \underset{M_{\dag},\preceq_{\dag}}{\Rightarrow} (\theta_{\dag}(L'),\lambda')$.
  Secondly, we verify the statement (s2).
  Consider the case of $L'=L_1$.
  In the case of $\tau=t_1$ or $t_1 t_2$, by Proposition \ref{prop:haji-ga-not-preceq} there exists no integer sequence $\lambda' \in {\mathbb{Z}_{\geq 1}}^+$ such that $(\theta_{\dag}(L_1),<4,5,6>) \underset{M_{\dag},\preceq_{\dag}}{\Rightarrow} (\tau,\lambda')$, and
  in the case of $\tau \neq t_1, t_1 t_2$, $\tau$ is not a proper prefix of $\theta_{\dag}(L')$, and thus, the statement (s2) holds.
  Consider the case of $L'=L_2$.
  In the case of $\tau=t_1$, by Proposition \ref{prop:haji-ga-not-preceq} there exists no integer sequence $\lambda' \in {\mathbb{Z}_{\geq 1}}^+$ such that $(\theta_{\dag}(L_1),<4,5,6>) \underset{M_{\dag},\preceq_{\dag}}{\Rightarrow} (\tau,\lambda')$, and
  in the case of $\tau \neq t_1$, $\tau$ is not a proper prefix of $\theta_{\dag}(L')$, and thus, the statement (s2) holds.
  Consider the case of $L'=L_3$ or $L_4$.
  Any $\tau \in {\Gamma_{\dag}}^+$ is not a proper prefix of $\theta_{\dag}(L')$, and thus, the statement (s2) holds.
  Therefore, $\lambda_{L_1,a^{15}} = <4,5,6>$ satisfies the statements (s1) and (s2).
  
  In the same way, the statements (s1) and (s2) are satisfied by giving an integer sequence
  $\lambda_{L_2,a^{15}} = <8,7>$ for $L_2$ and $a^{15} \in L_2$,
  $\lambda_{L_2,b^{7}} = <4,4>$ for $L_2$ and $b^{7} \in L_2$,
  $\lambda_{L_3,c^{5}} = <5>$ for $L_3$ and $c^{5} \in L_3$,
  $\lambda_{L_4,c^{5}} = <5>$ for $L_4$ and $c^{5} \in L_4$, and
  $\lambda_{L_4,d^{4}} = <4>$ for $L_4$ and $d^{4} \in L_4$.
  \halmos
\end{rei}

\begin{rei}
  \label{rei:jouken-1-rei2}
  We consider $\mathcal{L}_{\dag}$, $M_{\dag}$, and $\Gamma_{\dag}$ defined in Example \ref{rei:jouken-1}.
  Let $\preceq_{\dag}' = \{ (t,t) \mid t \in \Gamma \} \cup \{ (t_1,t_2),(t_2,t_3),(t_1,t_3) \}$ be a partial order over ${\Gamma_{\dag}}$,
  Then, we can show that there exist no injection $\theta_{\dag}':\mathcal{L}_{\dag} \to {\Gamma_{\dag}}^+$
  and no class $\delta_{\mathcal{L}_{\dag}}' = \{ \delta_{L_1}',\ldots,\delta_{L_4}' \}$ of Boolean functions $\delta_L'$ from $L \; (\in \mathcal{L}_{\dag})$ to $\{0,1\}$ satisfying the condition (C) with respect to $\mathcal{L}_{\dag}$, $M_{\dag}$, and $({\Gamma_{\dag}},\preceq_{\dag}')$.
  We can show this by contradiction.
  Assume that an injection $\theta_{\dag}'$
  and a class $\delta_{\mathcal{L}_{\dag}}' = \{ \delta_{L_1}',\ldots,\delta_{L_4}' \}$ of Boolean functions
  satisfy the condition (C) with respect to $\mathcal{L}_{\dag}$, $M_{\dag}$, and $({\Gamma_{\dag}},\preceq_{\dag}')$.
  By (c2) of Definition \ref{defi:condition-syn}, for any $L \in \mathcal{L}_{\dag}$ and for any $w \in L$,
  there exists $\lambda \in \mathrm{Div}(|\theta_{\dag}'(L)|,|w|+\delta_{L}'(w),[4,8])$ such that (s1) and (s2) hold.
  We write such $\lambda$ as $\lambda_{L,w}$.
  We have $\lambda_{L_4,d^4} \in \mathrm{Div}(|\theta_{\dag}'(L_4)|,|d^4|+\delta_{L_4}'(d^4),[4,8])$, which implies $|\theta_{\dag}'(L_4)|=1$.
  In the same way, we have $|\theta_{\dag}'(L_2)|=2$ and $|\theta_{\dag}'(L_3)|=1$.
  Let $L=L_3$ and $L'=L_2$ in the statement (s2).
  If $\theta_{\dag}'(L_3)=t_1$ holds, for any $\tau \in \Gamma_{\dag}$, we have $(\theta_{\dag}'(L_3),\lambda_{L_3,c^5}) \underset{M_{\dag},\preceq_{\dag}}{\Rightarrow} (\tau,\lambda_{L_3,c^5})$.
  Then, (s2) implies that $t_1$, $t_2$, and $t_3$ are not a proper prefix of $\theta_{\dag}'(L_2)$, which is a contradiction.
  Therefore, we have $\theta_{\dag}'(L_3)=t_2$ or $t_3$. In the same way, $\theta_{\dag}'(L_4)=t_2$ or $t_3$.
  By (s1), we have $\theta_{\dag}'(L_3)=t_2$ and $\theta_{\dag}'(L_4)=t_3$.
  Then, $\lambda' \in {\mathbb{Z}_{\geq 1}}^+$ satisfying $(\theta_{\dag}'(L_2),\lambda_{L_2,b^7}) \underset{M_{\dag},\preceq_{\dag}}{\Rightarrow} (\theta_{\dag}'(L_4),\lambda')$ always exists, which contradicts (s1) since $b^7 \not\in L_4$ holds.
  \halmos
\end{rei}

\begin{prop}
  \label{prop:preceq-subset-c}

  Let $\preceq_1$ and $\preceq_2$ be partial orders over $\Gamma$.
  Assume that $\preceq_1 \; \subseteq \; \preceq_2$ holds.
  If $\theta$ and $\delta_{\mathcal{L}}$ satisfy the condition (C) with respect to $\mathcal{L}$, $M$, and $(\Gamma,\preceq_2)$,
  then, $\theta$ and $\delta_{\mathcal{L}}$ satisfy the condition (C) with respect to $\mathcal{L}$, $M$, and $(\Gamma,\preceq_1)$.
  
\end{prop}


\begin{proof}

  Assume that $\theta$ and $\delta_{\mathcal{L}}$ satisfy the condition (C) with respect to $\mathcal{L}$, $M=(\mu_t,\mu_b)$, and $(\Gamma,\preceq_2)$.
  Then, by Definition \ref{defi:condition-syn}, we have that
  \begin{align*}
    \mbox{(c1)} \; & \mathrm{Alph}(\theta(\mathcal{L})) = \Gamma \mbox{ holds}, \mbox{ and}\\
    \mbox{(c2$_{\preceq_2}$)} \; & \mbox{for any } L \in \mathcal{L}, \mbox{ for any } w \in L, \mbox{ there exists } \lambda \in \mathrm{Div}(|\theta(L)|,|w|+\delta_{L}(w),\mu_b) \\
    & \mbox{such that the following (s1$_{\preceq_2}$) and (s2$_{\preceq_2}$) hold}: \\
    & \mbox{(s1$_{\preceq_2}$)} \; \forall L' \in \mathcal{L} \left( \left( \exists \lambda' \in {\mathbb{Z}_{\geq 1}}^+ \left( (\theta(L),\lambda) \underset{M,\preceq_2}{\Rightarrow} (\theta(L'),\lambda') \right) \right) \mbox{ implies } w \in L' \right), \\
    & \mbox{(s2$_{\preceq_2}$)} \; \forall L' \in \mathcal{L}, \forall \tau \in \Gamma^+
    \left(
    \begin{aligned}
      & \left( \exists \lambda' \in {\mathbb{Z}_{\geq 1}}^+ \left( (\theta(L),\lambda) \underset{M,\preceq_2}{\Rightarrow} (\tau,\lambda') \right) \right) \mbox{ implies } \\
      & \quad \left( \tau \mbox{ is not a proper prefix of } \theta(L') \right)
    \end{aligned}
    \right).
  \end{align*}
  By Proposition \ref{prop:preceq-subset-rightarrow-subset}, we have $\underset{M,\preceq_1}{\Rightarrow} \; \subseteq \; \underset{M,\preceq_2}{\Rightarrow}$.
  Therefore, by (c2$_{\preceq_2}$), we have that
  \begin{align*}
    \mbox{(c2$_{\preceq_1}$)} \; & \mbox{for any } L \in \mathcal{L}, \mbox{ for any } w \in L, \mbox{ there exists } \lambda \in \mathrm{Div}(|\theta(L)|,|w|+\delta_{L}(w),\mu_b) \\
    & \mbox{such that the following (s1$_{\preceq_1}$) and (s2$_{\preceq_1}$) hold}: \\
    & \mbox{(s1$_{\preceq_1}$)} \; \forall L' \in \mathcal{L} \left( \left( \exists \lambda' \in {\mathbb{Z}_{\geq 1}}^+ \left( (\theta(L),\lambda) \underset{M,\preceq_1}{\Rightarrow} (\theta(L'),\lambda') \right) \right) \mbox{ implies } w \in L' \right), \\
    & \mbox{(s2$_{\preceq_1}$)} \; \forall L' \in \mathcal{L}, \forall \tau \in \Gamma^+
    \left(
    \begin{aligned}
      & \left( \exists \lambda' \in {\mathbb{Z}_{\geq 1}}^+ \left( (\theta(L),\lambda) \underset{M,\preceq_1}{\Rightarrow} (\tau,\lambda') \right) \right) \mbox{ implies } \\
      & \quad \left( \tau \mbox{ is not a proper prefix of } \theta(L') \right)
    \end{aligned}
    \right).
  \end{align*}
  Therefore, since (c1) and (c2$_{\preceq_1}$) hold,
  $\theta$ and $\delta_{\mathcal{L}}$ satisfy the condition (C) with respect to $\mathcal{L}$, $M$, and $(\Gamma,\preceq_1)$.
\end{proof}

\subsection{Construction of $H_*$ and $C_*$}
\label{subsec:construction-of-h-c-syn}

Let $\mathcal{L}$ be a finite class of non-empty finite languages over a finite alphabet $\Sigma$,
$M=(\mu_t,\mu_b)$ be a GC, and
$(\Gamma,\preceq)$ be a partially ordered finite control alphabet.
Assume that there exist
an injection $\theta:\mathcal{L} \to \Gamma^+$ and
a class $\delta_{\mathcal{L}}=\{ \delta_L \mid L \in \mathcal{L} \}$ of Boolean functions $\delta_L$ from $L \; (\in \mathcal{L})$ to $\{ 0,1 \}$
satisfying the condition (C) with respect to $\mathcal{L}$, $M$, and $(\Gamma,\preceq)$.
Using the definitions of $\theta$ and $\delta_{\mathcal{L}}$,
we will give constructions of an RLUB $H_*=(G,M)$ and a control system $C_*=(\Gamma,\phi,T)$
such that
$\preceq = \preceq_{C_*}$ holds and
$H_*$ and $C_*$ synchronously generate $\mathcal{L}$ in the erasing mode.

We first give a construction of an RLUB $H_*=(G,M)$.
For each $L \in \mathcal{L}$ and $w \in L$, we will define a set $P$ of production rules of $G$
so that the length of derivation of $w \in L$ should be $|w| + \delta_{L}(w)$.
In other words, in the case of $\delta_{L}(w)=0$, the number of production rules to generate $w \in L$ is $|w|$, and
in the case of $\delta_{L}(w)=1$, the number of production rules to generate $w \in L$ is $|w|+1$.
The additional derivation step of length $1$ in the case of $\delta_{L}(w)=1$
is achieved by the use of $\epsilon$ rule of the form $A \to \epsilon$.
Furthermore, we construct the production rules so that all nonterminal symbols used in them are different from each other except for the start symbol $S$.
Formally, $H_*$ is defined as follows:
\begin{align*}
  \mbox{(D1)} \quad & H_*=((V,\Sigma,S,P),M), \\
  & \quad V = \{ S \} \cup \left( \bigcup_{L \in \mathcal{L}} V(L) \right) \cup \{ Z^{(L,w)} \mid L \in \mathcal{L}, w \in L \}, \displaybreak[1] \\
  & \quad\quad V(L) = \bigcup_{w \in L} V(L,w) \quad (\mbox{for any } L \in \mathcal{L}), \displaybreak[1] \\
  & \quad\quad\quad V(L,w) = \left\{
  \begin{aligned}
    & \{ A_i^{(L,w)} \mid 1 \leq i \leq |w|-1 \} \quad (\mbox{if } \delta_{L}(w)=0) \\
    & \{ A_i^{(L,w)} \mid 1 \leq i \leq |w| \} \quad (\mbox{if } \delta_{L}(w)=1)
  \end{aligned}
  \right. \\
  & \quad\quad\quad\quad\quad\quad\quad\quad\quad\quad\quad\quad (\mbox{for any } L \in \mathcal{L} \mbox{ and } w \in L), \displaybreak[1] \\
  & \quad P = P_0 \cup (\bigcup_{L \in \mathcal{L}} P(L)), \displaybreak[1] \\
  & \quad\quad P_0 = \{ Z^{(L,w)} \to \epsilon \mid L \in \mathcal{L}, w \in L \}, \displaybreak[1] \\
  & \quad\quad P(L) = \bigcup_{w \in L} P(L,w) \quad (\mbox{for any } L \in \mathcal{L}), \displaybreak[1] \\
  & \quad\quad\quad P(L,w) = \left\{
  \begin{aligned}
    & \{ S \to w[1] A_1^{(L,w)},
    A_1^{(L,w)} \to w[2] A_2^{(L,w)},
    A_2^{(L,w)} \to w[3] A_3^{(L,w)}, \ldots , \\
    & \quad A_{|w|-2} \to w[|w|-1] A_{|w|-1}^{(L,w)},
    A_{|w|-1}^{(L,w)} \to w[|w|] \} \quad (\mbox{if } \delta_{L}(w)=0) \\
    & \{ S \to w[1] A_1^{(L,w)},
    A_1^{(L,w)} \to w[2] A_2^{(L,w)},
    A_2^{(L,w)} \to w[3] A_3^{(L,w)}, \ldots , \\
    & \quad A_{|w|-1} \to w[|w|] A_{|w|}^{(L,w)},
    A_{|w|}^{(L,w)} \to \epsilon \} \quad (\mbox{if } \delta_{L}(w)=1) \\
  \end{aligned}
  \right. \\
  & \quad\quad\quad\quad\quad\quad\quad\quad\quad\quad\quad\quad (\mbox{for any } L \in \mathcal{L} \mbox{ and } w \in L \mbox{ such that } |w|+\delta_{L}(w) \geq 2 ), \\
  & \quad\quad\quad P(L,w) = \{ S \to w \} \quad (\mbox{for any } L \in \mathcal{L} \mbox{ and } w \in L \mbox{ such that } |w|+\delta_{L}(w) = 1).
\end{align*}
Note that since $\mathcal{L}$ is a finite class of non-empty finite languages, the sets $V$ and $P$ are finite sets, which implies that $H_*$ is well-defined.

Let $L \in \mathcal{L}$ and $w \in L$.
For any integer $i$ with $1 \leq i \leq |P(L,w)|$,
by $r_i^{(L,w)}$, we denote the element of $P(L,w)$ which is applied the $i$-th in the process of deriving $w$ using $P(L,w)$.
We define $R(L,w)$ as the sequence from $r_1^{(L,w)}$ to $r_{|P(L,w)|}^{(L,w)}$,
that is, we define $R(L,w) \overset{\mathrm{def}}{=} r_1^{(L,w)} \cdots r_{|P(L,w)|}^{(L,w)}$.
Note that we have $S \underset{R(L,w)}{\Rightarrow} w$.
In addition,
by $r_0^{(L,w)}$, we denote the rule $Z^{(L,w)} \to \epsilon$ in $P_0$.
The rules in $P_0$ are used in the definition (D2) mentioned later.

\begin{rei}
  \label{rei:example-of-h-syn}

  We give an example construction of $H_*$.
  We consider $\mathcal{L}_{\dag}$, $M_{\dag}$, $\Gamma_{\dag}$, $\preceq_{\dag}$, $\theta_{\dag}$, and $\delta_{\mathcal{L}_{\dag}}$ defined in Example \ref{rei:jouken-1}.
  By the definition (D1),
  we define $H_*=((V,\Sigma,S,P),M_{\dag})$, where
  $V = \{ S \} \cup V(L_1) \cup \cdots \cup V(L_4)$ and
  $P = P(L_1) \cup \cdots \cup P(L_4)$.
  For example, the definition (D1) leads to $V(L_2)=V(L_2,a^{15}) \cup V(L_2,b^{7})$ and $P(L_2) = P(L_2,a^{15}) \cup P(L_2,b^{7})$, where
  $V(L_2,a^{15})=\{ A_{1}^{(L_2,a^{15})},\ldots,A_{14}^{(L_2,a^{15})} \}$,
  $V(L_2,b^{7})=\{ A_{1}^{(L_2,b^{7})},\ldots,A_{7}^{(L_2,b^{7})} \}$,
  $P(L_2,a^{15}) = \{ S \to a A_{1}^{(L_2,a^{15})},A_{1}^{(L_2,a^{15})} \to a A_{2}^{(L_2,a^{15})}, \ldots, A_{13}^{(L_2,a^{15})} \to a A_{14}^{(L_2,a^{15})}, A_{14}^{(L_2,a^{15})} \to a \}$, and
  $P(L_2,b^{7}) = \{ S \to b A_{1}^{(L_2,b^{7})},A_{1}^{(L_2,b^{7})} \to b A_{2}^{(L_2,b^{7})}, \ldots, A_{6}^{(L_2,b^{7})} \to b A_{7}^{(L_2,b^{7})},A_{7}^{(L_2,b^{7})} \to \epsilon \}$.
  For any $L \in \mathcal{L}_{\dag}$ and any $w \in L$,
  the set $P(L,w)$ of production rules is defined in order to derive $w$ at $|w|+\delta_{L}(w)$ steps.

  The notations $r_i^{(L,w)}$ are used to specify each element of $P(L,w)$.
  For example, for $L_2 \in \mathcal{L}_{\dag}$ and $b^{7} \in L_2$, we have
  $r_{1}^{(L_2,b^{7})} : S \to b A_{1}^{(L_2,b^{7})}$,
  $r_{2}^{(L_2,b^{7})} : A_{1}^{(L_2,b^{7})} \to b A_{2}^{(L_2,b^{7})}$,
  $\ldots$,
  $r_{7}^{(L_2,b^{7})} : A_{6}^{(L_2,b^{7})} \to b A_{7}^{(L_2,b^{7})}$,
  $r_{8}^{(L_2,b^{7})} : A_{7}^{(L_2,b^{7})} \to \epsilon$.
  Moreover, we have $R(L_2,b^{7}) = r_{1}^{(L_2,b^{7})} \cdots r_{8}^{(L_2,b^{7})}$ and
  $S \underset{R(L_2,b^{7})}{\Rightarrow} b^{7}$.
  \halmos
\end{rei}

For $L \in \mathcal{L}$ and $w \in L$, we divide the set $P(L,w)$ into disjoint subsets
based on the integer sequence $\lambda \in {\mathbb{Z}_{\geq 0}}^+$ such that $\sum \lambda=|P(L,w)|$.
In the case of $\lambda = <l_1,\ldots,l_k>$ such that $\lambda \neq <1>$, we divide $P(L,w)$ into $k$ pieces of disjoint subsets
so that
the 1st subset $P(L,w)_{\lambda}^{(1)}$ contains the 1st $l_1$ rules,
the 2nd subset $P(L,w)_{\lambda}^{(2)}$ contains the 2nd $l_2$ rules,
$\ldots$,
the $k$-th subset $P(L,w)_{\lambda}^{(k)}$ contains the last $k$-th $l_k$ rules
according to the application order in the derivation process of $w$ in $L$.
In the case of $\lambda=<1>$, we put $r_0^{(L,w)}$ as well as $r_1^{(L,w)}$ into $P(L,w)_{\lambda}^{(1)}$.
Formally, we divide $P(L,w)$ as follows:
\begin{align*}
  \mbox{(D2)} \quad P(L,w)_{\lambda}^{(i)} = \left\{
  \begin{aligned}
    & \{ r_0^{(L,w)}, r_1^{(L,w)} \} \quad (\mbox{if } \lambda=<1>)\\
    & \{ r_{k}^{(L,w)} \mid \sum_{j=1}^{i-1} \lambda[j] + 1 \leq k \leq \sum_{j=1}^{i} \lambda[j] \} \quad (\mbox{otherwise}).\\
  \end{aligned}
  \right. \quad (1 \leq i \leq |\lambda|)
\end{align*}

\begin{defi}
  \label{defi:lambda-l-w}
  Since $\theta$ and $\delta_{\mathcal{L}}$ satisfy the condition (C) with respect to $\mathcal{L}$, $M=(\mu_t,\mu_b)$, and $(\Gamma,\preceq)$,
  we have that
  for any $L \in \mathcal{L}$ and for any $w \in L$, there exists $\lambda \in \mathrm{Div}(|\theta(L)|,|w|+\delta_{L}(w),\mu_b)$ such that the statements (s1) and (s2) of Definition \ref{defi:condition-syn} hold.
  We write such $\lambda$ as $\lambda_{L,w}$.
  \halmos
\end{defi}

It is straightforward to show the following Fact \ref{fact:lambda-l-w-elements-geq-1}, Fact \ref{fact:lambda-l-w-theta-l-length-equal}, Fact \ref{fact:p-l-w-neq-emptyset}, and Fact \ref{fact:different-p-l-w-dont-have-same-element}.
\begin{fact}
  \label{fact:lambda-l-w-elements-geq-1}
  If $0 \not\in \mu_b$ holds,
  we have $\lambda_{L,w} \in {\mathbb{Z}_{\geq 1}}^+$ for any $L \in \mathcal{L}$ and $w \in L$.
\end{fact}
\begin{fact}
  \label{fact:lambda-l-w-theta-l-length-equal}
  We have $|\lambda_{L,w}|=|\theta(L)|$ for any $L \in \mathcal{L}$ and $w \in L$.
\end{fact}
\begin{fact}
  \label{fact:p-l-w-neq-emptyset}
  If $0 \not\in \mu_b$ holds,
  we have that
  for any $L \in \mathcal{L}$, $w \in L$, and integer $i$ such that $1 \leq i \leq |\lambda_{L,w}|$,
  we have $P(L,w)_{\lambda_{L,w}}^{(i)}$ includes some rule other than $S \to x \; (x \in \Sigma \cup \{ \epsilon \})$.
\end{fact}
\begin{fact}
  \label{fact:different-p-l-w-dont-have-same-element}
  Assume that $0 \not\in \mu_b$ holds.
  For any $L \in \mathcal{L}$, $L' \in \mathcal{L}$, $w \in L$, $w' \in L'$, integer $i$ such that $1 \leq i \leq |\lambda_{L,w}|$, and integer $i'$ such that $1 \leq i '\leq |\lambda_{L',w'}|$,
  we have that $P(L,w)_{\lambda_{L,w}}^{(i)} \cap P(L',w')_{\lambda_{L',w'}}^{(i')} \cap(P - \{ S \to x \mid x \in \Sigma \cup \{ \epsilon \} \}) \neq \emptyset$ implies $L=L'$, $w=w'$, and $i=i'$.
\end{fact}

We will next give the construction of $C_*$ using the definition of $\theta$ and $\lambda_{L,w} \mbox{'s} \; (L \in \mathcal{L}, w \in L)$.
For each $t \in \Gamma$, $\phi(t)$ is constructed in the following manner:
starting from the initialization $\phi(t) = \emptyset$,
for each $L \in \mathcal{L}$, $w \in L$, and $i$ with $1 \leq i \leq |\lambda_{L,w}|$,
we put the elements of $P(L,w)_{\lambda_{L,w}}^{(i)}$ into $\phi(t)$ if $\theta(L)[i] \preceq t$ holds.
Formally, $C_*$ is defined as follows\footnote{
Note that for a class $\mathcal{X}$ of sets, we define $\bigcup \mathcal{X} \overset{\mathrm{def}}{\equiv} \bigcup_{X \in \mathcal{X}} X$.
}:
\begin{align*}
  \mbox{(D3)} \quad & C_*=(\Gamma,\phi,T), \\
  & \quad \phi(t) = \bigcup \{ P(L,w)_{\lambda_{L,w}}^{(i)} \mid L \in \mathcal{L}, w \in L, i \in [1,|\theta(L)|], \mbox{ and }\theta(L)[i] \preceq t \} \\
  & \quad\quad\quad\quad\quad\quad\quad\quad\quad\quad\quad\quad \quad\quad\quad\quad\quad\quad\quad\quad\quad\quad\quad\quad (\mbox{for any } t \in \Gamma), \\
  & \quad T = \theta(\mathcal{L}).
\end{align*}
Note that by Fact \ref{fact:lambda-l-w-theta-l-length-equal}, $[1,|\theta(L)|] = [1,|\lambda_{L,w}|]$ holds for any $L \in \mathcal{L}$ and $w \in L$.

We will give important remarks Remark \ref{rem:relation-r-phi-lambda}, Remark \ref{rem:preceq-preceq-c-equal}, and Remark \ref{rem:phi-of-Cstar-is-injective}
about the construction of $H_*$ and $C_*$.
In order to make the discussion clear, we recall the following setting once again.

Let $\mathcal{L}$ be a finite class of non-empty finite languages over a finite alphabet $\Sigma$,
$M=(\mu_t,\mu_b)$ be a GC, and
$(\Gamma,\preceq)$ be a partially ordered finite control alphabet.
Assume that there exist
an injection $\theta : \mathcal{L} \to \Gamma^+$ and
a class $\delta_{\mathcal{L}}=\{ \delta_L \mid L \in \mathcal{L} \}$ of Boolean functions $\delta_L$ from $L \; (\in \mathcal{L})$ to $\{ 0,1 \}$
satisfying the condition (C).
For $L \in \mathcal{L}$ and $w \in L$,
by $\lambda_{L,w}$, we denote an integer sequence $\lambda$ satisfying the statements (s1) and (s2) of Definition \ref{defi:condition-syn}.
Then, by using these $\lambda_{L,w}$'s, we construct an RLUB $H_*=(G,M)$ based on (D1) and a control system $C_*=(\Gamma,\phi,T)$ based on (D2) and (D3).

\begin{rem}
  \label{rem:relation-r-phi-lambda}
  Assume that $0 \not\in \mu_b$ holds.
  Let $L \in \mathcal{L}$, $w \in L$, $j \in [1,|P(L,w)|]$, and $t \in \Gamma$.
  We have that $r_j^{(L,w)} \in \phi(t)$ holds if and only if $\theta(L)[(\lambda_{L,w},j)] \preceq t$ holds.
\end{rem}
\begin{proof}
  By Fact \ref{fact:lambda-l-w-elements-geq-1},
  we have $\lambda_{L,w} \in {\mathbb{Z}_{\geq 1}}^+$.
  Then,
  by the definition (D2) and Fact \ref{fact:lambda}, we have that
  \begin{align}
    \label{eq:rem:relation-r-phi-lambda:relation-r-lambda}
    \mbox{for any } i \in [1,|\lambda_{L,w}|], \; \;r_j^{(L,w)} \in P(L,w)_{\lambda_{L,w}}^{(i)} \mbox{holds if and only if } i=(\lambda_{L,w},j) \mbox{ holds}.
  \end{align}
  
  Assume that $r_j^{(L,w)} \in \phi(t)$ holds.
  By the definition (D3), we have that there exists an integer $i' \in [1,|\theta(L)|]$ such that $r_j^{(L,w)} \in P(L,w)_{\lambda_{L,w}}^{(i')}$ and $\theta(L)[i'] \preceq t$ hold.
  By (\ref{eq:rem:relation-r-phi-lambda:relation-r-lambda}), we have $i'=(\lambda_{L,w},j)$.
  Therefore, $\theta(L)[i'] \preceq t$ implies $\theta(L)[(\lambda_{L,w},j)] \preceq t$.
  We obtain the only if direction.

  Assume that $\theta(L)[(\lambda_{L,w},j)] \preceq t$ holds.
  By the definition (D3),
  we have $P(L,w)_{\lambda_{L,w}}^{((\lambda_{L,w},j))} \subseteq \phi(t)$.
  By (\ref{eq:rem:relation-r-phi-lambda:relation-r-lambda}), we have $r_j^{(L,w)} \in P(L,w)_{\lambda_{L,w}}^{((\lambda_{L,w},j))}$.
  Therefore, $r_j^{(L,w)} \in \phi(t)$ holds.
  We obtain the if direction.
\end{proof}

\begin{rem}
  \label{rem:preceq-preceq-c-equal}
  Assume that $0 \not\in \mu_b$ holds.
  Then,
  $\preceq = \preceq_{C_*}$ holds.
\end{rem}


\begin{proof}

  We will show the following claim (\ref{eq:rem:preceq-preceq-c-equal:only-if-and-only-if}):
  \begin{align}
    \label{eq:rem:preceq-preceq-c-equal:only-if-and-only-if}
    \mbox{for any } t_1 \mbox{ and } t_2 \mbox{ in } \Gamma, \; t_1 \preceq t_2 \mbox{ holds if and only if } \phi(t_1) \subseteq \phi(t_2) \mbox{ holds}.
  \end{align}

  Let $t_1$ and $t_2$ be any elements in $\Gamma$.
  
  Assume that $t_1 \preceq t_2$ holds.
  By the definition (D3), we have
  \begin{align*}
    & \phi(t_1) = \bigcup \{ P(L,w)_{\lambda_{L,w}}^{(i)} \mid L \in \mathcal{L}, w \in L, i \in [1,|\theta(L)|], \mbox{ and }\theta(L)[i] \preceq t_1 \}, \mbox{ and} \\
    & \phi(t_2) = \bigcup \{ P(L,w)_{\lambda_{L,w}}^{(i)} \mid L \in \mathcal{L}, w \in L, i \in [1,|\theta(L)|], \mbox{ and }\theta(L)[i] \preceq t_2 \}.
  \end{align*}
  Therefore, $\phi(t_1) \subseteq \phi(t_2)$ holds.
  Thus, we obtain the only if direction of (\ref{eq:rem:preceq-preceq-c-equal:only-if-and-only-if}).
  
  Assume that $\phi(t_1) \subseteq \phi(t_2)$ holds.
  Since $\mathrm{Alph}(\theta(\mathcal{L})) = \Gamma$ holds by (c1) of Definition \ref{defi:condition-syn},
  there exist $L_* \in \mathcal{L}$ and $i_* \in [1,|\theta(L_*)|]$
  such that $\theta(L_*)[i_*]=t_1$ holds.
  Let $w_* \in L_*$.
  By the assumption $\phi(t_1) \subseteq \phi(t_2)$,
  we have $P(L_*,w_*)_{\lambda_{L_*,w_*}}^{(i_*)} \subseteq \phi(t_2)$.
  Since by Fact \ref{fact:p-l-w-neq-emptyset}, $P(L_*,w_*)_{\lambda_{L_*,w_*}}^{(i_*)}$ includes some rule other than $S \to x \; (x \in \Sigma \cup \{ \epsilon \})$,
  there exist $L' \in \mathcal{L}$, $w' \in L'$, and $i' \in [1,|\theta(L')|]$
  such that $\theta(L')[i'] \preceq t_2$ and $P(L_*,w_*)_{\lambda_{L_*,w_*}}^{(i_*)} \cap P(L',w')_{\lambda_{L',w'}}^{(i')} \cap(P - \{ S \to x \mid x \in \Sigma \cup \{ \epsilon \} \}) \neq \emptyset$ hold.
  Then, by Fact \ref{fact:different-p-l-w-dont-have-same-element},
  we have $L_*=L'$, $w_*=w'$, and $i_*=i'$.
  Therefore, we have $t_1 = \theta(L_*)[i_*] = \theta(L')[i'] \preceq t_2$.
  Thus, we obtain the if direction of (\ref{eq:rem:preceq-preceq-c-equal:only-if-and-only-if}).

  Since $\phi(t_1) \subseteq \phi(t_2) \Leftrightarrow t_1 \preceq_{C_*} t_2$ holds
  by the definition of $\preceq_{C_*}$ in section \ref{sec:rlub-and-its-control},
  we have $t_1 \preceq t_2 \Leftrightarrow t_1 \preceq_{C_*} t_2$.
  Thus, we have $\preceq = \preceq_{C_*}$.
\end{proof}

\begin{rem}
  \label{rem:phi-of-Cstar-is-injective}
  Assume that $0 \not\in \mu_b$ holds.
  Then, the control function $\phi$ of $C_*$ is an injection.
\end{rem}
\begin{proof}
  Since $0 \not\in \mu_b$ holds,
  we have the claim (\ref{eq:rem:preceq-preceq-c-equal:only-if-and-only-if}) of Remark \ref{rem:preceq-preceq-c-equal}.
  Let $t_1,t_2 \in \Gamma$ such that $\phi(t_1)=\phi(t_2)$ holds.
  By the claim (\ref{eq:rem:preceq-preceq-c-equal:only-if-and-only-if}) of Remark \ref{rem:preceq-preceq-c-equal},
  $\phi(t_1) \subseteq \phi(t_2)$ and $\phi(t_2) \subseteq \phi(t_1)$ imply $t_1 \preceq t_2$ and $t_2 \preceq t_1$.
  Thus, since $\preceq$ is antisymmetric, we have $t_1 = t_2$, which implies that $\phi$ is an injection.
\end{proof}

\begin{rei}
  \label{rei:example-of-c-syn}

  We consider $\mathcal{L}_{\dag}$, $M_{\dag}$, $\Gamma_{\dag}$, $\preceq_{\dag}$, $\theta_{\dag}$, and $\delta_{\mathcal{L}_{\dag}}$ defined in Example \ref{rei:jouken-1}
  and consider $H_*$ defined in Example \ref{rei:example-of-h-syn}.
  For any $L \in \mathcal{L}_{\dag}$ and $w \in L$, we divide $P(L,w)$ using $\lambda_{L,w}$'s
  which were found during the verification steps of the condition (C) (see Example \ref{rei:jouken-1}).
  For example, by the definition (D2), we divide $P(L_1,a^{15})$ using $\lambda_{L_1,a^{15}}=<4,5,6>$ into the following three subsets:
  $P(L_1,a^{15})_{<4,5,6>}^{(1)} = \{ r_{1}^{(L_1,a^{15})},\ldots,r_{4}^{(L_1,a^{15})} \}$,
  $P(L_1,a^{15})_{<4,5,6>}^{(2)} = \{ r_{5}^{(L_1,a^{15})},\ldots,r_{9}^{(L_1,a^{15})} \}$, and
  $P(L_1,a^{15})_{<4,5,6>}^{(3)} = \{ r_{10}^{(L_1,a^{15})},\ldots,r_{15}^{(L_1,a^{15})} \}$.

  By the definition (D3) based on $\theta_{\dag}$ and $\lambda_{L,w}$'s,
  we define $C_*=(\Gamma,\phi,T)$, where
  $\phi(t_1)=P(L_1,a^{15})_{<4,5,6>}^{(1)} \cup P(L_2,a^{15})_{<8,7>}^{(1)} \cup P(L_2,b^{7})_{<4,4>}^{(1)}$,
  $\phi(t_2)=P(L_1,a^{15})_{<4,5,6>}^{(2)} \cup P(L_3,c^{5})_{<5>}^{(1)}$, and
  $\phi(t_3)=P(L_1,a^{15})_{<4,5,6>}^{(2)} \cup P(L_1,a^{15})_{<4,5,6>}^{(3)} \cup P(L_2,a^{15})_{<8,7>}^{(2)} \cup P(L_2,b^{7})_{<4,4>}^{(2)} \cup P(L_3,c^{5})_{<5>}^{(1)} \cup P(L_4,c^{5})_{<5>}^{(1)} \cup P(L_4,d^{4})_{<4>}^{(1)}$, and
  $T = \theta_{\dag}(\mathcal{L}_{\dag}) \;\; (= \{ t_1 t_2 t_3, t_1 t_3, t_2, t_3 \})$.
  Note that $\phi(t_2) \subseteq \phi(t_3)$ holds and
  that $\preceq = \preceq_{C_*}$ holds.
  \halmos
\end{rei}

\subsection{Necessary and sufficient conditions for the controlled generation of RLUBs}
\label{subsec:proof-characterization-syn}

We first show that
the existence of $\theta$ and $\delta_{\mathcal{L}}$ satisfying the condition (C) is a sufficient condition for the controlled synchronous generation of RLUBs in the erasing mode.
\begin{theo}
  \label{theo:sufficient-syn}

  Let $\mathcal{L}$ be a finite class of non-empty finite languages over a finite alphabet $\Sigma$,
  $M=(\mu_t,\mu_b)$ be a GC with $0 \not\in \mu_t$, and
  $(\Gamma,\preceq)$ be a partially ordered finite control alphabet.
  Assume that
  there exist
  an injection $\theta:\mathcal{L} \to \Gamma^+$ and
  a class $\delta_{\mathcal{L}}=\{ \delta_L \mid L \in \mathcal{L} \}$ of Boolean functions $\delta_L$ from $L \; (\in \mathcal{L})$ to $\{ 0,1 \}$
  satisfying the condition (C) with respect to $\mathcal{L}$, $M$, and $(\Gamma,\preceq)$.
  Then,
  there exist
  an RLUB $H=(G,M)$ and a control system $C=(\Gamma,\phi,T)$ for $H$
  such that
  (i) $\mathrm{Alph}(T)=\Gamma$ holds,
  (ii) $\preceq = \preceq_{C}$ holds,
  (iii) $\mathcal{L}_e^{\mathrm{syn}}(H,C)=\mathcal{L}$ holds, and
  (iv) for any $\tau_1, \tau_2 \in T$, $\tau_1 \neq \tau_2$ implies $L_e^{\mathrm{syn}}(H,C,\tau_1) \neq L_e^{\mathrm{syn}}(H,C,\tau_2)$.

\end{theo}


\begin{proof}

  Assume that there exist
  an injection $\theta:\mathcal{L} \to \Gamma^+$ and
  a class $\delta_{\mathcal{L}}=\{ \delta_L \mid L \in \mathcal{L} \}$ of Boolean functions $\delta_L$ from $L \; (\in \mathcal{L})$ to $\{ 0,1 \}$
  satisfying the condition (C) with respect to $\mathcal{L}$, $M$, and $(\Gamma,\preceq)$.

  Since $\theta$ and $\delta_{\mathcal{L}}$ satisfy the condition (C) with respect to $\mathcal{L}$, $M$, and $(\Gamma,\preceq)$,
  we have that
  for any $L \in \mathcal{L}$ and for any $w \in L$, there exists $\lambda \in \mathrm{Div}(|\theta(L)|,|w|+\delta_{L}(w),\mu_b)$ such that the statements (s1) and (s2) of Definition \ref{defi:condition-syn} hold, where we should recall that the statements (s1) and (s2) are as follows:
  \begin{align*}
    & \mbox{(s1)} \; \forall L' \in \mathcal{L} \left( \left( \exists \lambda' \in {\mathbb{Z}_{\geq 1}}^+ \left( (\theta(L),\lambda) \underset{M,\preceq}{\Rightarrow} (\theta(L'),\lambda') \right) \right) \mbox{ implies } w \in L' \right), \\
    & \mbox{(s2)} \; \forall L' \in \mathcal{L}, \forall \tau \in \Gamma^+
    \left(
    \begin{aligned}
      & \left( \exists \lambda' \in {\mathbb{Z}_{\geq 1}}^+ \left( (\theta(L),\lambda) \underset{M,\preceq}{\Rightarrow} (\tau,\lambda') \right) \right) \mbox{ implies } \\
      & \quad \left( \tau \mbox{ is not a proper prefix of } \theta(L') \right)
    \end{aligned}
    \right).
  \end{align*}
  We write such $\lambda$ as $\lambda_{L,w}$.

  Let $H_*=((V,\Sigma,S,P),M)$ be an RLUB defined by the definition (D1) using $\mathcal{L}$.
  Since $\mathcal{L}$ is a finite class of non-empty finite languages,
  the construction of $H_*$ leads to that the sets $V$ and $P$ of $H_*$ are finite,
  and thus, $H_*$ is well-defined.
  Let $C_*=(\Gamma,\phi,T)$ be a control system for $H_*$ defined by the definition (D3), where we use integer sequences $\lambda_{L,w}$'s and the definition (D2).
  
  Since $T=\theta(\mathcal{L})$ holds by the definition (D3)
  and the statement (c1) of Definition \ref{defi:condition-syn} holds,
  we have that (i) $\mathrm{Alph}(T)=\Gamma$ holds.

  By Remark \ref{rem:preceq-preceq-c-equal},
  (ii) $\preceq = \preceq_{C_*}$ holds.

  Let $L \in \mathcal{L}$ and $\theta(L)=t_1 \cdots t_n$ for some $n$ in $\mathbb{Z}_{\geq 1}$ and some $t_i$'s in $\Gamma$.
  Let $X_i \subseteq (V \cup \Sigma)^*$ for each $i \in [1,n]$ such that
  \begin{align}
    \label{eq:theo:sufficient-syn:derivation-detail-upper}
    \{ S \}
    \underset{{\phi(t_1)}^{{\mu}_{t}}}{\Rightarrow^e} X_1
    \underset{{\phi(t_2)}^{{\mu}_{t}}}{\Rightarrow^e}
    \cdots
    \underset{{\phi(t_n)}^{{\mu}_{t}}}{\Rightarrow^e} X_n.
  \end{align}
  Let $Y_i \subseteq (V \cup \Sigma)^*$ for each $i \in [1,n]$ such that
  \begin{align}
    \label{eq:theo:sufficient-syn:derivation-detail-lower}
    \{ S \}
    \underset{{\phi(t_1)}^{{\mu}_{b}}}{\Rightarrow^e} Y_1
    \underset{{\phi(t_2)}^{{\mu}_{b}}}{\Rightarrow^e}
    \cdots
    \underset{{\phi(t_n)}^{{\mu}_{b}}}{\Rightarrow^e} Y_n .
  \end{align}
  By the definition (D3), for any $w \in L$, we have $P(L,w)_{\lambda_{L,w}}^{(i)} \subseteq \phi(\theta(L)[i]) = \phi(t_i)$ for each $i \in [1,n]$.
  Moreover, by the definition (D2), we have $|P(L,w)_{\lambda_{L,w}}^{(i)}| = \lambda_{L,w}[i] \in \mu_b \subseteq \mu_t$.
  Therefore, at each stage $i$ of (\ref{eq:theo:sufficient-syn:derivation-detail-upper}) (and (\ref{eq:theo:sufficient-syn:derivation-detail-lower}), respectively),
  we can apply the rules  in $P(L,w)_{\lambda_{L,w}}^{(i)}$ using the behavior $\phi(t_i)^{\mu_t}$ (and $\phi(t_i)^{\mu_b}$, respectively).
  Thus, we have that
  \begin{align}
    \label{eq:theo:sufficient-syn:sono1}
    \mbox{for any } w \in L, \mbox{ it holds that }w \in X_n \mbox{ and } w \in Y_n .
  \end{align}
  
  We will show the following claim (\ref{eq:theo:sufficient-syn:claim-star}):
  \begin{align}
    \label{eq:theo:sufficient-syn:claim-star}
    \left.
    \begin{aligned}
      & \mbox{for any } w'' \in \Sigma^* \mbox{ and } k \in [1,n], \mbox{ we have that} \\
      & \; w'' \in X_k \mbox{ implies that there exist } L'' \in \mathcal{L} \mbox{ with } w'' \in L'' \mbox{ and } \lambda \in {\mathbb{Z}_{\geq 1}}^+ \\
      & \; \mbox{satisfying } (\theta(L''),\lambda_{L'',w''}) \underset{M,\preceq}{\Rightarrow} (\theta(L)[1,k],\lambda).
    \end{aligned}
    \right\}
  \end{align}
  Let $w'' \in \Sigma^*$ and $k \in [1,n]$.
  Assume that $w'' \in X_k$ holds.
  By the definition (D1), $w''$ is generated using the rules of $P(L'',w'')$ for some $L'' \in \mathcal{L}$ with $w'' \in L''$.
  Then, by (\ref{eq:theo:sufficient-syn:derivation-detail-upper}),
  there exist behaviors $\alpha_1, \ldots, \alpha_k \in {P(L'',w'')}^*$
  such that
  $\alpha_1 \cdots \alpha_k = R(L'',w'')$ and
  $\alpha_1 \in {\phi(t_1)}^{{\mu}_{t}}$, $\ldots$, $\alpha_k \in {\phi(t_k)}^{{\mu}_{t}}$.
  Let $\lambda$ be an integer sequence such that $\lambda=<|\alpha_1|,\ldots,|\alpha_k|>$.
  By $0 \not\in \mu_t$, we have $ \lambda \in {\mathbb{Z}_{\geq 1}}^+$.
  We will show that $(\theta(L''),\lambda_{L'',w''}) \underset{M,\preceq}{\Rightarrow} (\theta(L)[1,k],\lambda)$ holds.
  Since $\lambda_{L'',w''} \in \mathrm{Div}(|\theta(L'')|,|w''|+\delta_{L''}(w''),\mu_b)$ holds,
  we have $\sum \lambda_{L'',w''} = |w''|+\delta_{L''}(w'')$.
  By the definition (D1), we have $|w''|+\delta_{L''}(w'') = |P(L'',w'')|$.
  Then, $\sum \lambda_{L'',w''} = |P(L'',w'')|$ holds.
  Moreover, $\sum \lambda = |\alpha_1| + \cdots + |\alpha_k| = |R(L'',w'')| = |P(L'',w'')|$,
  which implies (A1) of Definition \ref{def:Rightarrow-M-preceq}.
  Let $m$ be an integer such that $\sum \lambda = \sum \lambda_{L'',w''} = m$ holds.
  Let $i$ be any integer in $[1,k]$ and $j$ be any integer with $\sum_{p=1}^{i-1} \lambda[p] + 1 \leq j \leq \sum_{p=1}^{i} \lambda[p]$.
  By the definition of $\lambda$, $\alpha_i = R(L'',w'')[\sum_{p=1}^{i-1} \lambda[p] + 1, \sum_{p=1}^{i} \lambda[p]]$ holds\footnote{
  We should recall that $R(L'',w'')$ is a string consisting of rules as symbols and
  $R(L'',w'')[i,j]$ with $1 \leq i \leq j \leq |R(L'',w'')|$ is the substring starting from the $i$-th letter and ending at the $j$-th letter of $R(L'',w'')$.
  }.
  Then,
  since $\alpha_i \in {\phi(t_i)}^{\mu_t}$ holds,
  we have $r_{j}^{(L'',w'')} \in \phi(t_i)$, which implies $r_{j}^{(L'',w'')} \in \phi(\theta(L)[i])$.
  Therefore, by Remark \ref{rem:relation-r-phi-lambda}, we have $\theta(L'')[(\lambda_{L'',w''},j)] \preceq \theta(L)[i]$.
  Since $i=(\lambda,j)$ holds by Fact \ref{fact:lambda}, we have $\theta(L'')[(\lambda_{L'',w''},j)] \preceq \theta(L)[(\lambda,j)]$,
  which impiles $\theta(L'')[(\lambda_{L'',w''},j)] \preceq \theta(L)[1,k][(\lambda,j)]$.
  Since $j$ can be any integer in $[1,m]$, we obtain (A3) of Definition \ref{def:Rightarrow-M-preceq}.
  We have that $|\alpha_i| \in \mu_t$ holds for any $i \in [1,k]$.
  Moreover, we have $|\lambda|=k$.
  Therefore, $\lambda \in {\mu_t}^{k} = {\mu_t}^{|\theta(L)[1,k]|}$,
  which implies (A2) of Definition \ref{def:Rightarrow-M-preceq}.
  Thus, since we obtain (A1), (A2), and (A3) of Definition \ref{def:Rightarrow-M-preceq},
  we have $(\theta(L''),\lambda_{L'',w''}) \underset{M,\preceq}{\Rightarrow} (\theta(L)[1,k],\lambda)$,
  which completes the proof of the claim (\ref{eq:theo:sufficient-syn:claim-star}).

  Let $w''' \in \Sigma^*$.
  Assume that $w''' \in X_n$ holds.
  By the claim (\ref{eq:theo:sufficient-syn:claim-star}), there exist $L''' \in \mathcal{L}$ with $w''' \in L'''$ and $\lambda \in {\mathbb{Z}_{\geq 1}}^+$
  satisfying $(\theta(L'''),\lambda_{L''',w'''}) \underset{M,\preceq}{\Rightarrow} (\theta(L)[1,n],\lambda)$.
  Note that $\theta(L)[1,n] = \theta(L)$ holds.
  Then, by (s1) of Definition \ref{defi:condition-syn}, we have $w''' \in L$.
  Thus, we have that
  \begin{align}
    \label{eq:theo:sufficient-syn:sono2}
    \mbox{for any } w''' \in \Sigma^*, w''' \in X_n \mbox{ implies } w''' \in L.
  \end{align}
  By (\ref{eq:theo:sufficient-syn:sono2}), since $w''' \in Y_n$ implies $w''' \in X_n$, we have that
  \begin{align}
    \label{eq:theo:sufficient-syn:sono3}
    \mbox{for any } w''' \in \Sigma^*, w''' \in Y_n \mbox{ implies } w''' \in L.
  \end{align}

  By (\ref{eq:theo:sufficient-syn:sono1}), (\ref{eq:theo:sufficient-syn:sono2}), and (\ref{eq:theo:sufficient-syn:sono3}),
  we have $X_n \cap \Sigma^* = Y_n \cap \Sigma^* = L$.

  Let $w'''' \in \Sigma^*$.
  Let $k$ be an integer with $1 \leq k < n$.
  Assume that $w'''' \in X_k$ holds.
  By (\ref{eq:theo:sufficient-syn:claim-star}), there exist $L'''' \in \mathcal{L}$ with $w'''' \in L''''$ and $\lambda \in {\mathbb{Z}_{\geq 1}}^+$
  satisfying $(\theta(L''''),\lambda_{L'''',w''''}) \underset{M,\preceq}{\Rightarrow} (\theta(L)[1,k],\lambda)$.
  Then, by (s2) of Definition \ref{defi:condition-syn}, we have that $\theta(L)[1,k]$ is not a proper prefix of $\theta(L)$,
  which is a contradiction since $k < n$ holds.
  Therefore, we have that $w'''' \not\in X_k$ holds.
  Since $Y_i \subseteq X_i$ holds for $i \in [1,n]$ by $\mu_b \subseteq \mu_t$, we have that $w'''' \not\in Y_k$ holds.
  Therefore, we have $(X_1 \cup \cdots \cup X_{n-1}) \cap \Sigma^* = \emptyset$ and $(Y_1 \cup \cdots \cup Y_{n-1}) \cap \Sigma^* = \emptyset$.
  Thus, since $X_n \cap \Sigma^* = Y_n \cap \Sigma^* = L$ holds,
  we have $L_e^{\mathrm{syn}}(((V,\Sigma,S,P),(\mu_t,\mu_t)),C_*,\theta(L)) = L$ and $ L_e^{\mathrm{syn}}(((V,\Sigma,S,P),(\mu_b,\mu_b)),C_*,\theta(L)) = L$.
  By Theorem \ref{theo:generation-kimoto}, we have $L_e^{\mathrm{syn}}(H_*,C_*,\theta(L)) = L$.
  Therefore, we have (iii) $\mathcal{L}_e^{\mathrm{syn}}(H_*,C_*) = \mathcal{L}$.
  Moreover, for any $\theta(L_1),\theta(L_2) \in T$, we have that (iv) $\theta(L_1) \neq \theta(L_2)$ implies $L_e^{\mathrm{syn}}(H_*,C_*,\theta(L_1)) \neq L_e^{\mathrm{syn}}(H_*,C_*,\theta(L_2))$ since $\theta$ is an injection.
\end{proof}

We next show that
the existence of $\theta$ and $\delta_{\mathcal{L}}$ satisfying the condition (C) is a necessary condition for the controlled synchronous generation of RLUBs in the erasing mode.
\begin{theo}
  \label{theo:necessary-syn}

  Let $\mathcal{L}$ be a finite class of non-empty finite languages over a finite alphabet $\Sigma$,
  $M=(\mu_t ,\mu_b)$ be a GC, and
  $(\Gamma,\preceq)$ be a partially ordered finite control alphabet.
  Assume that
  there exist
  an RLUB $H=(G,M)$ and a control system $C=(\Gamma,\phi,T)$ for $H$
  such that
  (i) $\mathrm{Alph}(T)=\Gamma$ holds,
  (ii) $\preceq = \preceq_{C}$ holds,
  (iii) $\mathcal{L}_e^{\mathrm{syn}}(H,C)=\mathcal{L}$ holds, and
  (iv) for any $\tau_1, \tau_2 \in T$, $\tau_1 \neq \tau_2$ implies $L_e^{\mathrm{syn}}(H,C,\tau_1) \neq L_e^{\mathrm{syn}}(H,C,\tau_2)$.
  Then,
  there exist
  an injection $\theta:\mathcal{L} \to \Gamma^+$ and
  a class $\delta_{\mathcal{L}}=\{ \delta_L \mid L \in \mathcal{L} \}$ of Boolean functions $\delta_L$ from $L \; (\in \mathcal{L})$ to $\{0,1\}$
  satisfying the condition (C) with respect to $\mathcal{L}$, $M$, and $(\Gamma,\preceq)$.

\end{theo}


\begin{proof}

  Assume that an RLUB $H=(G,M)$ with $G=(V,\Sigma,S,P)$ and a control system $C=(\Gamma,\phi,T)$ for $H$ satisfy (i), (ii), (iii), and (iv).
  
  We first define $\theta$.
  By (iii) and (iv), for any $L \in \mathcal{L}$, there exists a unique $\tau_L \in T$ such that $L_e^{\mathrm{syn}}(H,C,\tau_L)=L$.
  We define $\theta(L) = \tau_L$ for any $L \in \mathcal{L}$.  
  Let $L_1,L_2 \in \mathcal{L}$.
  Then, $\theta(L_1)=\theta(L_2)$ implies 
  $L_1 = L_e^{\mathrm{syn}}(H,C,\tau_{L_1}) = L_e^{\mathrm{syn}}(H,C,\theta(L_1)) = L_e^{\mathrm{syn}}(H,C,\theta(L_2)) = L_e^{\mathrm{syn}}(H,C,\tau_{L_2}) = L_2$.
  Therefore, $\theta$ is an injection.

  We next define $\delta_{\mathcal{L}}$.
  Let $L$ be any language in $\mathcal{L}$ and $w$ be any string in $L$.
  Let $\theta(L) = t_1 \cdots t_n$ for some $n$ in ${\mathbb{Z}_{\geq 1}}$ and some $t_i$'s in $\Gamma$.
  By the definition of $\theta$, there exist $X_i \subseteq (V \cup \Sigma)^*$ for each $i \in [1,n]$ such that
  \begin{align}
    \label{eq:theo:necessary-syn:generation-detail}
    \{ S \}
    \underset{{\phi(t_1)}^{{\mu}_{b}}}{\Rightarrow^e} X_1
    \underset{{\phi(t_2)}^{{\mu}_{b}}}{\Rightarrow^e}
    \cdots
    \underset{{\phi(t_n)}^{{\mu}_{b}}}{\Rightarrow^e} X_n
    \quad \mbox{and} \quad w \in X_n .
  \end{align}
  Then, there exist $\alpha_1,\ldots,\alpha_n \in P^*$
  such that
  $S \underset{\alpha_1 \cdots \alpha_n}{\Rightarrow} w$ and
  $\alpha_1 \in {\phi(t_1)}^{{\mu}_{b}},\ldots,\alpha_n \in {\phi(t_n)}^{{\mu}_{b}}$.
  We define $\delta_{\mathcal{L}}$ as follows:
  \begin{align*}
    \delta_L(w) = \left\{
    \begin{aligned}
      & 1 \quad (\mbox{if } |\alpha_1 \cdots \alpha_n| = |w|+1) \\
      & 0 \quad (\mbox{if } |\alpha_1 \cdots \alpha_n| = |w|).
    \end{aligned}
    \right.
  \end{align*}
  Note that the length of derivation of $w$ should be either $|w|$ or $|w|+1$ since $G$ is a right linear grammar.

  We next show that $\theta$ and $\delta_{\mathcal{L}}$ satisfy the condition (C) with respect to $\mathcal{L}$, $M$, and $(\Gamma,\preceq)$.
  By the definition of $\theta$, $\theta(\mathcal{L}) \subseteq T$ holds.
  By (iii) and (iv), $T \subseteq \theta(\mathcal{L})$ holds.
  Then, we have $\theta(\mathcal{L}) = T$.
  Since (i) holds, we have $\mathrm{Alph}(\theta(\mathcal{L})) = \mathrm{Alph}(T) = \Gamma$,
  which implies (c1) of Definition \ref{defi:condition-syn}.
  We next show that (c2) holds.
  In order to show that (c2) holds, we should prove the existence of $\lambda$ which satisfies (s1) and (s2).
  
  Let $L$ be any language in $\mathcal{L}$ and $w$ be any string in $L$.
  Such $w$ always exists since $L$ is a non-empty language.
  Let $\theta(L) = t_1 \cdots t_n$ for some $n$ in ${\mathbb{Z}_{\geq 1}}$ and some $t_i$'s in $\Gamma$.
  Recall the derivation process (\ref{eq:theo:necessary-syn:generation-detail}) for generating $w$ and production rule sequences $\alpha_1,\ldots,\alpha_n \in P^*$ used for defining $\delta_L(w)$.
  We define $\lambda=<|\alpha_1|,\ldots,|\alpha_n|>$.
  Then, we show that $\lambda$ is an element in $\mathrm{Div}(|\theta(L)|,|w|+\delta_{L}(w),\mu_b)$.
  We have $|\lambda| = |\theta(L)| \; (= n)$.
  By the definition of $\delta_{\mathcal{L}}$, we have $\sum \lambda = |\alpha_1 \cdots \alpha_n| = |w|+\delta_{L}(w)$.
  Moreover, we have $\lambda[1]=|\alpha_1| \in \mu_b$, $\ldots$, $\lambda[n]=|\alpha_n| \in \mu_b$.
  Therefore, we have $\lambda \in \mathrm{Div}(|\theta(L)|,|w|+\delta_{L}(w),\mu_b)$.
  We will show that (s1) and (s2) of Definition \ref{defi:condition-syn} are satisfied for this $\lambda$.

  We first show that (s1) are satisfied.
  Let $L'$ be any language in $\mathcal{L}$.
  Let $\lambda'$ be an integer sequence in ${\mathbb{Z}_{\geq 1}}^+$ such that $(\theta(L),\lambda) \underset{M,\preceq}{\Rightarrow} (\theta(L'),\lambda')$ holds.
  Since $\preceq = \preceq_C$ holds by (ii), we have $(\theta(L),\lambda) \underset{M,\preceq_C}{\Rightarrow} (\theta(L'),\lambda')$.
  By (A2) of Definition \ref{def:Rightarrow-M-preceq}, since $\lambda \in {\mu_b}^{|\theta(L)|}$ holds, we have $\lambda' \in {\mu_t}^{|\theta(L')|}$.
  Therefore, there exists $Z \subseteq (V \cup \Sigma)^*$
  such that $\{ S \} \overset{\lambda'}{\underset{{\phi(\theta(L'))}^{{\mu}_{t}}}{\Rightarrow^e}} Z$,
  where we should recall Definition \ref{defi:rightarrow-R-l}, \ref{defi:rightarrow-gamma-lambda}, and \ref{defi:phi-tau-mu}.
  Let $Y \subseteq (V \cup \Sigma)^*$ such that $\{ S \} \overset{\lambda}{\underset{{\phi(\theta(L))}^{{\mu}_{b}}}{\Rightarrow^e}} Y$.
  By (\ref{eq:theo:necessary-syn:generation-detail}) and the definition of $\lambda$, we have $w \in Y$.
  By Lemma \ref{lemm:binaryrelation-subset}, we have $Y \subseteq Z$.
  Therefore, we have $w \in Z$.
  Since $L_e(H,C,\theta(L'))$ is defined,
  there exist $W_i \subseteq (V \cup \Sigma)^*$ for each $i \in [1,|\theta(L')|]$ such that
  \begin{align*}
    \{ S \}
    \underset{{\phi(\theta(L')[1])}^{{\mu}_{t}}}{\Rightarrow^e} W_1
    \underset{{\phi(\theta(L')[2])}^{{\mu}_{t}}}{\Rightarrow^e}
    \cdots
    \underset{{\phi(\theta(L')[|\theta(L')|])}^{{\mu}_{t}}}{\Rightarrow^e} W_{|\theta(L')|}.
  \end{align*}
  Since $L_e(H,C,\theta(L')) = L'$ holds, we have  $W_{|\theta(L')|} \cap \Sigma^* = L'$.
  Then, by Remark \ref{rem:subset-rightarrow}, we have $Z \subseteq W_{|\theta(L')|}$.
  Therefore, we have $w \in Z \cap \Sigma^* \subseteq W_{|\theta(L')|} \cap \Sigma^* = L'$, which completes the proof of the statement (s1).

  We next show that (s2) are satisfied.
  Let $L'$ be any language in $\mathcal{L}$ and $\tau$ be any control sequence in $\Gamma^+$.
  Let $\lambda'$ be an integer sequence in ${\mathbb{Z}_{\geq 1}}^+$ such that $(\theta(L),\lambda) \underset{M,\preceq}{\Rightarrow} (\tau,\lambda')$ holds.
  We will prove by contradiction that $\tau$ is not a proper prefix of $\theta(L')$.
  Assume that $\tau$ is a proper prefix of $\theta(L')$.
  Since $\preceq = \preceq_C$ holds by (ii), we have $(\theta(L),\lambda) \underset{M,\preceq_C}{\Rightarrow} (\tau,\lambda')$.
  By (A2) of Definition \ref{def:Rightarrow-M-preceq}, since $\lambda \in {\mu_b}^{|\theta(L)|}$ holds, we have $\lambda' \in {\mu_t}^{|\tau|}$.
  Therefore, there exists $Z' \subseteq (V \cup \Sigma)^*$
  such that $\{ S \} \overset{\lambda'}{\underset{{\phi(\tau)}^{{\mu}_{t}}}{\Rightarrow^e}} Z'$.
  Let $Y' \subseteq (V \cup \Sigma)^*$ such that $\{ S \} \overset{\lambda}{\underset{{\phi(\theta(L))}^{{\mu}_{b}}}{\Rightarrow^e}} Y'$.
  By (\ref{eq:theo:necessary-syn:generation-detail}) and the definition of $\lambda$, we have $w \in Y'$.
  By Lemma \ref{lemm:binaryrelation-subset}, we have $Y' \subseteq Z'$.
  Therefore, we have $w \in Z'$.
  Since $L_e^{\mathrm{syn}}(H,C,\theta(L'))$ is defined,
  there exist $A_i \subseteq {(V \cup \Sigma)}^*$ for each $i \in [1,|\theta(L')|]$ such that
  \begin{align}
    \label{eq:theo:necessary-syn:theta-l-dash-process}
    \{ S \}
    \underset{{\phi(\theta(L')[1])}^{{\mu}_{t}}}{\Rightarrow^e} A_1
    \underset{{\phi(\theta(L')[2])}^{{\mu}_{t}}}{\Rightarrow^e}
    \cdots
    \underset{{\phi(\theta(L')[|\theta(L')|])}^{{\mu}_{t}}}{\Rightarrow^e} A_{|\theta(L')|}.
  \end{align}
  By the assumption that $\tau$ is a proper prefix of $\theta(L')$, we have $\theta(L') = \tau \tau'$ for some control sequence $\tau'$ in $\Gamma^+$.
  Therefore, by (\ref{eq:theo:necessary-syn:theta-l-dash-process}),
  there exists $A' \in \{ A_1,\ldots, A_{|\theta(L')|-1} \}$ such that
  \begin{align*}
    \{ S \}
    \underset{{\phi(\tau[1])}^{{\mu}_{t}}}{\Rightarrow^e}
    \cdots
    \underset{{\phi(\tau[|\tau|])}^{{\mu}_{t}}}{\Rightarrow^e} A'
    \underset{{\phi(\tau'[1])}^{{\mu}_{t}}}{\Rightarrow^e}
    \cdots
    \underset{{\phi(\tau'[|\tau'|])}^{{\mu}_{t}}}{\Rightarrow^e} A_{|\theta(L')|} .
  \end{align*}
  By Remark \ref{rem:subset-rightarrow}, we have $Z' \subseteq A'$, which implies $w \in A'$.
  Therefore, we have $w \in A' \cap \Sigma^*$.
  However, since $L_e^{\mathrm{syn}}(H,C,\theta(L'))$ is defined, we have $(A_1 \cup \cdots \cup A_{|\theta(L')|-1}) \cap \Sigma^* = \emptyset$
  and thus, $A' \cap \Sigma^* = \emptyset$.
  This is a contradiction.
  Therefore, we have that $\tau$ is not a proper prefix of $\theta(L')$, which completes the proof of the statement (s2).

  Thus, $\theta$ and $\delta_{\mathcal{L}}$ satisfy the condition (C) with respect to $\mathcal{L}$, $M$, and $(\Gamma,\preceq)$.
\end{proof}
\begin{rei}
  We consider $\mathcal{L}_{\dag}$, $M_{\dag}$, $\Gamma_{\dag}$, $\preceq_{\dag}$, $\theta_{\dag}$, and $\delta_{\mathcal{L}_{\dag}}$ defined in Example \ref{rei:jouken-1}.
   By Example \ref{rei:jouken-1}, this $\theta_{\dag}$ and $\delta_{\mathcal{L}_{\dag}}$ satisfy the condition (C) with respect to $\mathcal{L}_{\dag}$, $M_{\dag}$, and $(\Gamma_{\dag},\preceq_{\dag})$.  
  Then, by Theorem \ref{theo:sufficient-syn},
  there exist
  an RLUB $H=(G,M_{\dag})$ and a control system $C=(\Gamma_{\dag},\phi,T)$ for $H$
  such that
  (i) $\mathrm{Alph}(T)=\Gamma_{\dag}$ holds,
  (ii) $\preceq_{\dag} = \preceq_{C}$ holds,
  (iii) $\mathcal{L}_e(H,C)=\mathcal{L}_{\dag}$ holds, and
  (iv) for any $\tau_1, \tau_2 \in T$, $\tau_1 \neq \tau_2$ implies $L_e(H,C,\tau_1) \neq L_e(H,C,\tau_2)$.
  Actually,
  the RLUB $H_*$ defined in Example \ref{rei:example-of-h-syn} and
  the control system $C_*$ defined in Example \ref{rei:example-of-c-syn}
  satisfy (i),(ii),(iii),(iv).

  We consider the another partial order $\preceq_{\dag}'$ defined in Example \ref{rei:jouken-1-rei2}.
  By Example \ref{rei:jouken-1-rei2}, there exist no injection $\theta_{\dag}':\mathcal{L}_{\dag} \to {\Gamma_{\dag}}^+$
  and no class $\delta_{\mathcal{L}_{\dag}}' = \{ \delta_{L_1}',\ldots,\delta_{L_4}' \}$ of Boolean functions $\delta_L'$ from $L \; (\in \mathcal{L}_{\dag})$ to $\{0,1\}$ satisfying the condition (C) with respect to $\mathcal{L}_{\dag}$, $M_{\dag}$, and $({\Gamma_{\dag}},\preceq_{\dag}')$.
  Then, by Theorem \ref{theo:necessary-syn},
  there exist
  no RLUB $H=(G,M_{\dag})$ and no control system $C=(\Gamma_{\dag},\phi,T)$ for $H$
  such that
  (i) $\mathrm{Alph}(T)=\Gamma_{\dag}$ holds,
  (ii) $\preceq_{\dag}' = \preceq_{C}$ holds,
  (iii) $\mathcal{L}_e^{\mathrm{syn}}(H,C)=\mathcal{L}_{\dag}$ holds, and
  (iv) for any $\tau_1, \tau_2 \in T$, $\tau_1 \neq \tau_2$ implies $L_e^{\mathrm{syn}}(H,C,\tau_1) \neq L_e^{\mathrm{syn}}(H,C,\tau_2)$.
  \halmos
\end{rei}

Finally, we will give necessary and sufficient conditions for the case of
(possibly) non-synchronous controlled generation. 
The definition of synchronous controlled
generation requires the additional condition (r3) on page \pageref{page:condition_r3} in section \ref{sec:rlub-and-its-control} as well as the conditions (r1) and (r2).
The proofs of Theorem \ref{theo:sufficient-syn}
and Theorem \ref{theo:necessary-syn} say that
the condition (r3) directly corresponds to the statement (s2)
of the condition (C) in if and only if directions independently
of the other conditions (r1) and (r2). Therefore,
it is straightforward to see that the following modified
condition (C') obtained by removing statement (s2) from (C)
can contribute to the characterization of (possibly) non-synchronous
controlled generation of RLUBs.

\begin{defi}
  \label{defi:condition-not-syn}
  Let $\mathcal{L}$ be a finite class of non-empty finite languages over a finite alphabet $\Sigma$,
  $M=(\mu_t,\mu_b)$ be a GC, and
  $(\Gamma,\preceq)$ be a partially ordered finite control alphabet.
  Let $\theta$ be an injection from $\mathcal{L}$ to $\Gamma^+$ and
  $\delta_{\mathcal{L}}=\{ \delta_L \mid L \in \mathcal{L} \}$ be a class of Boolean functions $\delta_L$ from $L \; (\in \mathcal{L})$ to $\{ 0,1 \}$.
  We say that {\em $\theta$ and $\delta_{\mathcal{L}}$ satisfy the condition (C') with respect to $\mathcal{L}$, $M$, and $(\Gamma,\preceq)$} if the following (c1) and (c2') hold:
  \begin{align*}
    \mbox{(c1)} \; & \mathrm{Alph}(\theta(\mathcal{L})) = \Gamma \mbox{ holds}, \mbox{ and}\\
    \mbox{(c2')} \; & \mbox{for any } L \in \mathcal{L}, \mbox{ for any } w \in L, \mbox{ there exists } \lambda \in \mathrm{Div}(|\theta(L)|,|w|+\delta_{L}(w),\mu_b) \\
    & \mbox{such that the following (s1) holds}: \\
    & \mbox{(s1)} \; \forall L' \in \mathcal{L} \left( \left( \exists \lambda' \in {\mathbb{Z}_{\geq 1}}^+ \left( (\theta(L),\lambda) \underset{M,\preceq}{\Rightarrow} (\theta(L'),\lambda') \right) \right) \mbox{ implies } w \in L' \right).
  \end{align*}
  \halmos
\end{defi}

\begin{theo}
  \label{theo:sufficient-not-syn}

  Let $\mathcal{L}$ be a finite class of non-empty finite languages over a finite alphabet $\Sigma$,
  $M=(\mu_t,\mu_b)$ be a GC with $0 \not\in \mu_t$, and
  $(\Gamma,\preceq)$ be a partially ordered finite control alphabet.
  Assume that
  there exist
  an injection $\theta:\mathcal{L} \to \Gamma^+$ and
  a class $\delta_{\mathcal{L}}=\{ \delta_L \mid L \in \mathcal{L} \}$ of Boolean functions $\delta_L$ from $L \; (\in \mathcal{L})$ to $\{ 0,1 \}$
  satisfying the condition (C') with respect to $\mathcal{L}$, $M$, and $(\Gamma,\preceq)$.
  Then,
  there exist
  an RLUB $H=(G,M)$ and a control system $C=(\Gamma,\phi,T)$ for $H$
  such that
  (i) $\mathrm{Alph}(T)=\Gamma$ holds,
  (ii) $\preceq = \preceq_{C}$ holds,
  (iii) $\mathcal{L}_e(H,C)=\mathcal{L}$ holds, and
  (iv) for any $\tau_1, \tau_2 \in T$, $\tau_1 \neq \tau_2$ implies $L_e(H,C,\tau_1) \neq L_e(H,C,\tau_2)$.

\end{theo}

\begin{theo}
  \label{theo:necessary-not-syn}

  Let $\mathcal{L}$ be a finite class of non-empty finite languages over a finite alphabet $\Sigma$,
  $M=(\mu_t ,\mu_b)$ be a GC, and
  $(\Gamma,\preceq)$ be a partially ordered finite control alphabet.
  Assume that
  there exist
  an RLUB $H=(G,M)$ and a control system $C=(\Gamma,\phi,T)$ for $H$
  such that
  (i) $\mathrm{Alph}(T)=\Gamma$ holds,
  (ii) $\preceq = \preceq_{C}$ holds,
  (iii) $\mathcal{L}_e(H,C)=\mathcal{L}$ holds, and
  (iv) for any $\tau_1, \tau_2 \in T$, $\tau_1 \neq \tau_2$ implies $L_e(H,C,\tau_1) \neq L_e(H,C,\tau_2)$.
  Then,
  there exist
  an injection $\theta:\mathcal{L} \to \Gamma^+$ and
  a class $\delta_{\mathcal{L}}=\{ \delta_L \mid L \in \mathcal{L} \}$ of Boolean functions $\delta_L$ from $L \; (\in \mathcal{L})$ to $\{0,1\}$
  satisfying the condition (C') with respect to $\mathcal{L}$, $M$, and $(\Gamma,\preceq)$.

\end{theo}




\section{Conclusions and Future Works}
\label{sec:Conclusions}

This paper aimed to greatly strengthen the theoretical foundation
for controlled generation of RLUBs proposed in \cite{Kimoto,Kimoto2}.
We first introduced a partial order $\preceq_C$ over $\Gamma$
of a control system $C=(\Gamma,\phi,T)$,
which reflects the physical constraints of
control devices used in $C$.
Although we only considered the case that $\preceq_C$
is a {\em total} order
over $\Gamma$ in the previous works (\cite{Kimoto,Kimoto2}),
this paper made a detailed analysis on the language classes
generated by a control system $C$ such that $\preceq_C$ is
a {\em partial} order.
The goal of this paper was to answer to the question
informally explained as follows:
``given a finite class $\mathcal{L}$ of finite languages, a generative
condition $M$ and a partial order $\preceq$ over the control alphabet $\Gamma$,
answer whether there exist an RLUB $H$ using $M$
and its control system $C$ using $\Gamma$ such that $H$ and $C$ generate
$\mathcal{L}$ and $\preceq_C = \preceq$ holds.''
For this purpose, for any given $M$ and $\preceq$,
we introduce the important relation $\underset{M,\preceq}{\Rightarrow}$
over $\Phi(\Gamma)$. 
Using the relation $\underset{M,\preceq}{\Rightarrow}$,
under the assumption that $M=(\mu_t,\mu_b)$ satisfies
$0\not\in\mu_t$, we gave necessary conditions and
sufficient conditions to answer ``yes''
to the above question.

We have several problems which remain to be solved in the future works.
First, we could not succeed in removing the assumption
$0\not\in\mu_t$ of the generative condition $M=(\mu_t,\mu_b)$.
At this point, we do not have a good idea for
obtaining characterization theorems in the case of $0\in\mu_t$.
Another related unsolved problem is how to characterize
class of finite languages to be generated by
RLUBs and their control systems in the {\em remaining} mode.
Although we know some relationship between generative capacity
of RLUBs in the erasing mode and that in the remaining mode
(Theorem 1 in \cite{Kimoto}), it is not enough to reveal
the computational capability of RLUBs in the remaining mode
through the characterization in the case of erasing mode.
Finally, it is interesting to apply those characterization
theorems to reveal the hierarchy of generative capacity
of RLUBs $H$ and control systems $C$
with various physical constraints $\preceq_C$
imposed by control devices which we can use to implement $C$.
The obtained theorems in this paper
could help understand the computational capability
of the developed control devices such as temperature dependent
DNA devices, photo-responsive DNA devices, etc.

\section*{Acknowledgement}

This work was supported in part by a 
Grant-in-Aid for Scientific Research (B) (No. 19H04204) of 
Japan Society of Promotion of Science.


\end{document}